\documentclass[journal]{IEEEtran}

\usepackage{cite, url, array, calc, color}
\usepackage{algorithm, algorithmic}
\usepackage{graphics, epstopdf, graphicx, epsfig, psfrag, subfigure,}
\usepackage{amssymb, amsmath}
\usepackage{multicol,balance}

\newtheorem{theorem}{Theorem}[section]

\newtheorem{proposition}[theorem]{Proposition}

\newtheorem{definition}{Definition}
\newtheorem{example}{Example}

\newcommand{\ie}{{\em i.e., }}
\newcommand{\eg}{{\em e.g., }}

\hypersetup{breaklinks=true}

\begin{document}

\title{Optimal Source-Based Filtering of Malicious Traffic\\
\thanks{This work was supported by the NSF CyberTrust grant 0831530.}}
%\title{Optimal Filtering of Malicious IP Sources}
%\title{Optimal Filtering of Malicious Traffic \\ based on Source Address Prefixes}
%\title{Filtering IP Sources: Models and Algorithms}

\author{\authorblockN{Fabio Soldo, {\em IEEE Student Member},  Katerina Argyraki, {\em IEEE Member}, and Athina Markopoulou, {\em IEEE Member}}
\thanks{Fabio Soldo and Athina Markopoulou are with the Department of Electrical Engineering and Computer Science at the University of California, Irvine (email: \{fsoldo, athina\}@uci.edu).}
\thanks{Katerina Argyraki is with the School of Computer Science and Communication Sciences at EPFL, Switzerland (email: katerina.argyraki@epfl.ch).}
}

\maketitle

\begin{abstract}

In this paper, we consider the problem of blocking malicious traffic on the Internet, via source-based filtering.
 In particular, we consider filtering via access control lists (ACLs): these are already available at the routers today but are a scarce resource because they are stored in the expensive ternary content addressable memory (TCAM). Aggregation (by filtering source prefixes instead of individual IP addresses) helps reduce the number of filters, but comes also at the cost of blocking legitimate traffic originating from the filtered prefixes. We show how to optimally choose which source prefixes to  filter, for a variety of realistic attack scenarios and operators' policies. In each scenario, we design optimal, yet computationally efficient, algorithms. Using logs from {\tt Dshield.org}, we  evaluate the algorithms and demonstrate that they bring significant benefit in practice.

\end{abstract}

% \begin{keywords} \end{keywords}

\section{Introduction}
\label{sec:intro}

How can we protect our network infrastructure from malicious traffic, such as scanning,
malicious code propagation, spam, and distributed denial-of-service (DDoS) attacks?
These activities cause problems on a regular basis, ranging from simple annoyance
to severe financial, operational and political damage to companies, organizations
and critical infrastructure. In recent years, they have increased
in volume, sophistication, and automation, largely enabled by botnets, which are used
as the platform for launching these attacks.

%% Katerina's: positioning Filtering as a building block.
Protecting a victim (host or network) from malicious traffic is a hard problem that
requires the coordination of several complementary components, including
non-technical (\eg business and legal) and technical solutions (at the application and/or network level).
Filtering support from the network is a fundamental building block in this effort.
For example, an Internet service provider (ISP) may use filtering in response to an ongoing DDoS attack,
to block the DDoS traffic before it reaches its clients.
Another ISP may want to proactively identify and block traffic carrying malicious code before
it reaches and compromises vulnerable hosts in the first place. In either case,
filtering is a necessary operation that must be performed within the network.

Filtering capabilities are already available at routers today via access control lists (ACLs).
 ACLs enable a router to match a packet header against pre-defined rules and take pre-defined actions on the matching packets~\cite{acl1},
and they are currently used for enforcing a variety of policies, including infrastructure protection \cite{acl0}.
For the purpose of blocking malicious traffic, a filter is a simple ACL
rule that denies access to a source IP address or prefix.
To keep up with the high forwarding rates of modern routers, filtering is implemented
in hardware: ACLs are typically stored in ternary content addressable memory (TCAM), which
allows for parallel access and reduces the number of lookups per forwarded packet.
However, TCAM is more expensive and consumes more space and
power than conventional memory. The size and cost of TCAM puts a limit on the number of filters,
and this is not expected to change in the near future.\footnote{
%Enabling ACLs and policies does not decrease the switching or routing
%performance of the switch as long as the ACLs are fully loaded in the TCAM \cite{acl0}. If the TCAM is exhausted,
%the packets may be forwarded via the CPU path, which can decrease performance for those packets.
A router linecard or supervisor-engine card typically supports a single TCAM chip with tens of thousands
of entries. For example, the Cisco Catalyst 4500, a mid-range switch, provides a 64,000-entry TCAM to be
shared among all its interfaces (48- 384). Cisco 12000, a high-end router used at the Internet core, provides
20,000 entries that operate at line-speed per linecard (up to 4 Gigabit Ethernet interfaces). The Catalyst 6500
switch can fit 16K-32K patterns and 2K-4K masks in the TCAM. Depending on
how an ISP connects to its clients, each individual client can typically use only part of these ACLs, \ie a few
hundreds to a few thousands filters.
%The forwarding performance of a router does not degrade with the number of ACLs as long as these are loaded in TCAM \cite{acl0}.
}
With thousands or tens of thousands of filters per path, an ISP alone cannot hope to block
the currently witnessed attacks, not to mention attacks from multimillion-node botnets expected in the near future.

%% Athina's removed
%However, ACLs are a scarce resource because they are stored on the expensive TCAM.
%Indeed, TCAM chips come at the cost of lower storage density, higher power consumption,
%and eventually significant higher cost than SRAM or DRAM, which eventually limits  the number of ACLs
%stored on each TCAM chip. Although the exact number of filters available to each victim depends on the
%particular router's architecture\footnote{\textcolor{blue}{numbers...}} and on the operator's policy, the important %point is that there are less
%filters available (on the order of thousands or tens of thousands) than malicious sources (on the order
%of hundreds of thousands, given the magnitude of attacks today).

Consider the example shown in Fig.\ref{fig:overview}(a): an attacker commands a large number of compromised hosts to send
traffic to a victim $V$ (say a webserver), thus exhausting the resources of $V$ and preventing it from serving its legitimate clients.
 The ISP of $V$ tries to protect its client by blocking the attack at the gateway router $G$.
%\footnote{Identifying which sources are
%malicious and should be blocked is a difficult problem on its own right. However, it is orthogonal to the problem
%we study in this paper. Here, we consider
% a set of malicious IP sources as {\em input to our problem}, provided \eg by a blacklist or detection system, and %we strive to optimize the filter allocation to block this set of addresses.}
Ideally, $G$ should install one separate filter to block traffic from each attack source. However, there are typically fewer filters
than attack sources, hence aggregation is used, \ie a single filter (ACL) is used to block an entire source address prefix.
This has the desired effect of reducing the number of filters necessary to
block all attack traffic, but also the undesired effect of blocking legitimate traffic originating from the blocked prefixes
(we will call the damage that results from blocking legitimate traffic ``collateral damage'').
Therefore, filter selection can be viewed as an optimization problem that tries to block as many attack sources
with as little collateral damage as possible, given a limited number of filters. Furthermore, several measurement studies have demonstrated that malicious sources exhibit temporal and spatial clustering  \cite{uncleanness, clustering, mao2006analyzing, ramachandran2006understanding, venkataraman2007exploiting,xie2007dynamic,zhang-highly}, a feature that can be exploited by prefix-based filtering.

\begin{figure*}[t!]
\centering
\subfigure[Actual network.]
{\includegraphics[width=1.5in]{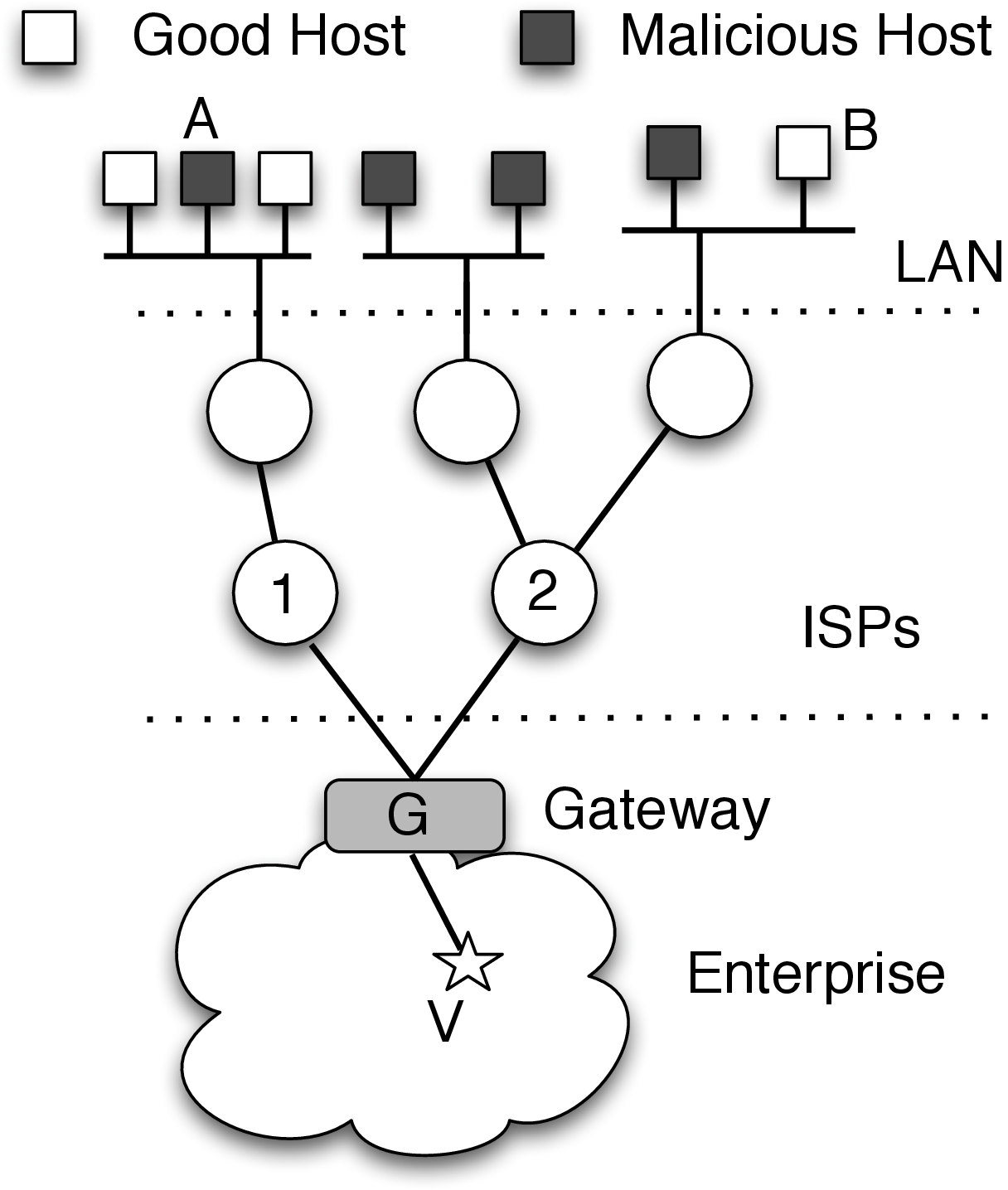}}
\hspace{50pt}
\subfigure[Hierarchy of source IP addresses and prefixes]
{\includegraphics[width=2.2in]{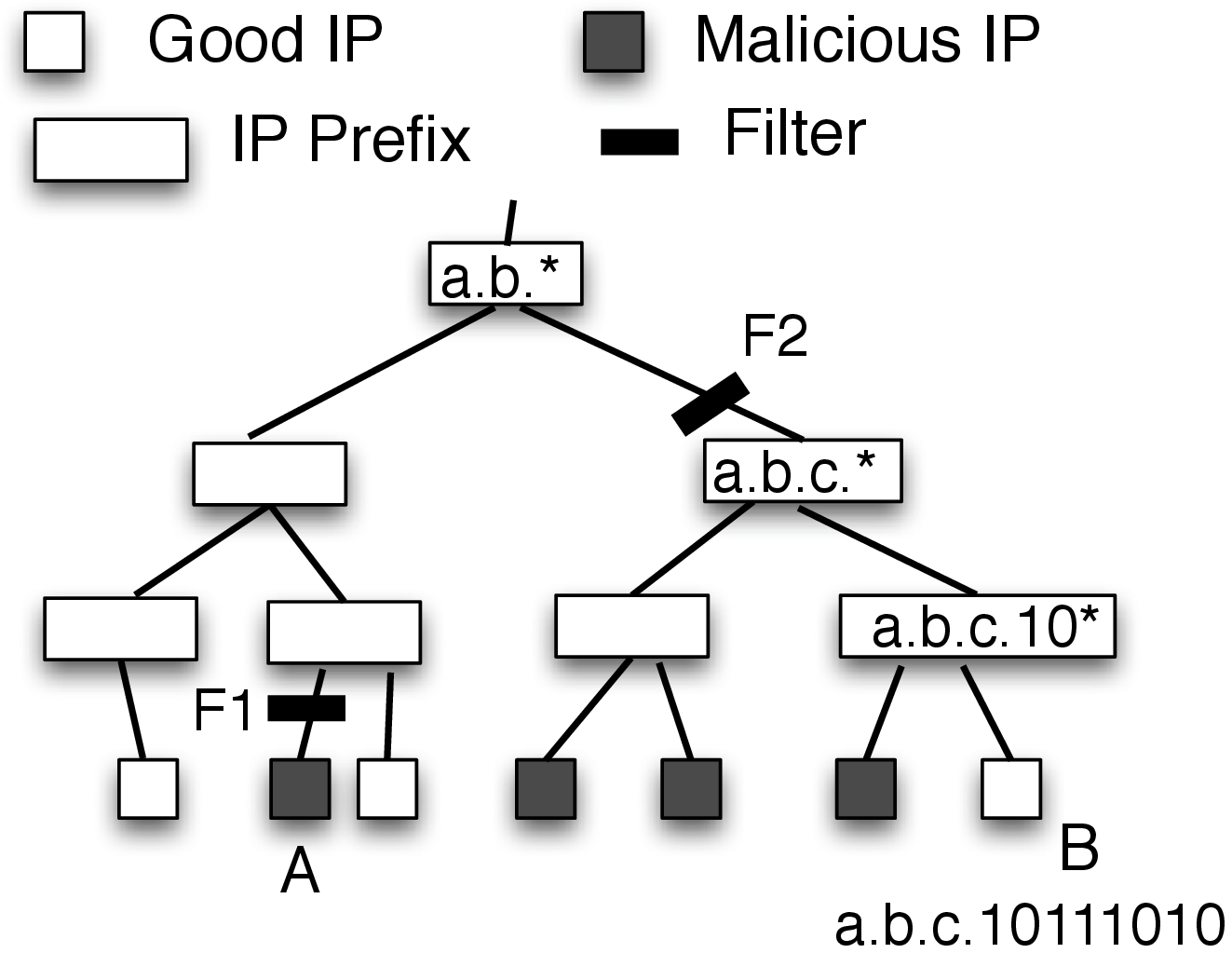}}
\caption{Example of a distributed attack. Let's assume that the gateway router $G$ has only two filters, $F1$ and $F2$, available to block malicious traffic and protect the victim $V$. It uses $F1$ to block a single malicious source address (A) and $F2$ to block the entire source prefix $a.b.c.*$, which contains 3 malicious sources but also one legitimate source (B). Therefore, the selection of filter $F2$ trades-off collateral damage (blocking B) for  reduction in the number of filters (from 3 to 1). We note that both filters, $F1$ and $F2$, are ACLs installed at the same router $G$.
\label{fig:overview}}
\end{figure*}

In this paper, we formulate a general framework for studying source prefix filtering as a
resource allocation problem. To the best of our knowledge, optimal filter selection has not
been explored so far, as most related work on filtering has focused on protocol and architectural aspects.
Within this framework, we formulate and solve five practical source-address filtering problems,
depending on the attack scenario and the operator's policy and constraints.
 Our contributions are twofold.
 On the theoretical side, filter selection optimization leads to novel variations of the multidimensional knapsack problem.
We exploit the special structure of each problem, and design optimal and computationally efficient algorithms.
 On the practical side, we provide a set of cost-efficient algorithms that
 can be used both by operators to block undesired traffic and by router manufacturers to optimize the
 use of TCAM and eventually the cost of routers.
%We would like to emphasize that we do not propose a novel architecture for dealing with attack traffic;
%instead, we optimize the use of an important mechanism
%that already exists in the Internet today and can be immediately used as a  building block in larger defense systems, as discussed in Section \ref{sec:assumptions}.
 We used logs from {\tt Dshield.org} to demonstrate that optimally selecting which source prefixes to filter brings significant benefits compared to non-optimized filtering or to generic clustering algorithms.

 The outline of the rest of the paper is as follows. In Section \ref{sec:formulation}, we formulate the general framework for optimal source prefix filtering. In Section \ref{sec:algorithms}, we study five specific problems that correspond to different attack scenarios and operator policies: blocking all addresses in a blacklist (BLOCK-ALL); blocking some addresses in a blacklist (BLOCK-SOME);  blocking all/some addresses in a time-varying blacklist (TIME-VARYING BLOCK-ALL/SOME); blocking flows during a DDoS flooding attack to meet bandwidth constraints (FLOODING); and distributed filtering across several routers during flooding (DIST-FLOODING). For each problem, we design an optimal, yet computationally efficient, algorithm to solve it. %for the optimal filter selection.
  %given an input blacklist, a constraint on the number of filters, and additional constraints depending on the particular problem.
  In Section \ref{sec:simulations}, we use data from {\tt Dshield.org} \cite{dshield} to evaluate the performance of our algorithms in realistic attack scenarios and we demonstrate that they bring significant benefit in practice.
Section \ref{sec:related} discusses related work and puts our work in perspective.
Section \ref{sec:conclusion} concludes the paper.% and outlines directions for future work.

\section{Problem Formulation and Framework}
\label{sec:formulation}

\subsection{Terminology and Notation}

% convetion:
% \mathcal --> sets
% capital characters --> constants
% small characters --> variables/parameters

Table \ref{tab:booktabs} summarizes our terminology and notation.

\subsubsection*{Source IP Addresses and Prefixes}
Every IPv4  address $ip$ is a 32-bit sequence.
% and can be thought of as an integer number in the range $[0, 2^{32}-1]$.
%$[i, j]$ denotes a set of consecutive addresses from address $ip$ to address $j$ (included).
We use standard IP/mask notation, \ie we write $p/l$  to indicate a prefix $p$ of length $l$ bits, where $p$ and $l$ can take values $l=0,1,...,32$ and $p=0,1,...,2^l-1$, respectively.
For brevity, when the meaning is obvious from the context, we simply write $p$ to indicate prefix $p/l$.
We write $ip \in p/l$ to indicate that address $ip$ is within the $2^{32-l}$ addresses covered by prefix $p/l$.
 %\ie in range $[p\cdot 2^{32-l}, (p+1)\cdot 2^{32-l}-1]$.
%Let us consider a set of consecutive IP addresses in that range, $\mathcal A = \{1, 2,\dots, m\}$,  in increasing order

\subsubsection*{Blacklists and Whitelists}
A blacklist ($\mathcal{BL}$) is a set of unique source IP addresses that send bad (undesired) traffic to the victim.
Similarly, a whitelist ($\mathcal{WL}$) is a set of unique source IP addresses that send good (legitimate) traffic to the victim.
An address may belong either to a blacklist (in which case we call it a ``bad'' address) or a to whitelist (in which case we call it a ``good'' address), but not to both.
We use $|\mathcal{BL}|$ and  $|\mathcal{WL}|$ to indicate the number of addresses in $\mathcal{BL}$ and $\mathcal{WL}$, respectively. For brevity, we also use $N=|\mathcal{BL}|$ for the number of addresses in the blacklist, which is the size of the most important input to our problem.

Each address $ip$ in a blacklist or a whitelist is assigned a weight $w_{ip}$, indicating the importance of that address.
If $ip$ is a bad address, we assign it a negative weight $w_{ip} \le 0$, which indicates the benefit from blocking $ip$;
if $ip$ is a good address, we assign it a positive weight $w_{ip}\geq 0$, which indicates the damage from blocking $ip$.
The higher the absolute value of the weight, the higher the benefit or damage and thus the preference to block the address or not

The weight $w_{ip}$ can have a different interpretation depending on the filtering problem.
For instance, it can represent the amount of bad/good traffic originating from the corresponding source address,
or it can express policy: depending on the amount of money gained/lost by the ISP when blocking source address $ip$,
an ISP operator can assign large positive weights to its important customers that should never be blocked,
or large negative weights to the worst attack sources that must definitely be blocked.
%For instance, if all good addresses are assigned the same weight $w_g$ and all bad addresses are assigned the same weight $-w_b$,
%then the ratio $\frac{w_g}{w_b}$ is a parameter that an ISP operator can tune to express how much she values preserving
%good traffic vs. blocking bad traffic.

Creating blacklists and whitelists (\ie identifying bad and good addresses and assigning appropriate weights to them)
is a difficult problem on its own right, but orthogonal to this work.
We assume that the blacklist $\mathcal{BL}$ is provided by another module (\eg an intrusion
detection system or historical data) as input to our problem. The sources of legitimate traffic are also assumed
known: \eg  web servers or ISPs typically keep historical data and know their customers. If it is not explicitly given, we take a conservative approach and define the whitelist $\mathcal{WL}$ to include all addresses that are not in $\mathcal{BL}$. %Further discussion on this assumption is provided in Section \ref{sec:related}

\subsubsection*{Filters}
We focus on filtering of source address prefixes.
In our context, a filter is an ACL rule that specifies that all packets with a source IP address in prefix $p/l$ should be blocked.
%In general, $l$ (left)  and $r$ (right) can be any  pair of IP addresses, either good or bad.
%The focus of this paper is on the commonly used {\em prefix based filtering}, where a single filter
%is used to block an entire source address prefix. In this case, $[l,r] \in \mathcal D$ where
%$\mathcal D$ is the set of continuous ranges of addresses (seen as non-negative integers) that represent a valid IPv4 prefixes. That is:
%\begin{itemize}
%\item $l, r \in \mathbf \{1,2...N\}$ %\footnote{\textcolor{blue}{do both l,r need to be bad??}}
%\item \label{eqn: prefixes domain}$\exists \ n\in\{0, 1,\dots, 32\}$ s.t. $(r-l)= 2^n-1$ and $ l = 0 ~(\mod 2^n)$
%\end{itemize}
%Throughout this paper, we only consider ranges that are valid prefixes: $[l,r] \in \mathcal D$.
%However, the framework can naturally incorporate other types of ranges for arbitrary pairs $(l,r)$.
$F_{max}$ is the maximum number of available filters, and it is given as input to our problem.
%$F \le F_{max}$ denotes the number of filters actually used.
  Filter optimization is meaningful only when $F_{max}$ is much smaller than the size of the blacklist $N=|\mathcal{BL}|$; otherwise the optimal would be to block every single bad address. $F_{max}<<N$ is indeed the case in practice due to the size and cost of the TCAM, as mentioned in Section~\ref{sec:intro}.
%As mentioned in the introduction and \eg in \cite{acl0,acl1}, the number of available ACLs in today's routers
%is on the order of tens of thousands (in total on the TCAM) or thousands (practically ending up been used for
%protection of a single client). At the same time, the size of attacks today is on the order hundreds of thousands
%and, enabled by botnets, growing in the order of millions.

The decision variable  $x_{p/l}\in \{1,0\}$ is $1$ if a filter is assigned to block prefix $p/l$; or $0$  otherwise.
A filter $p/l$ blocks all $2^{32-l}$ addresses in that range.
Hence, $b_{p/l}=|\sum_{ip\in p/l \cap \mathcal{BL}} w_{ip}|$ expresses the benefit from filter $p/l$,
whereas $g_{p/l}=\sum_{ip\in p/l \cap \mathcal{WL}} w_{ip}$ expresses the collateral damage it causes.
An effective filter should have a large benefit $b_{p/l}$ and low collateral damage  $g_{p/l}$.

%%% Do we ever use the notion of distance??
%Finally, we define the distance between two IPs, $i_1, \ i_2$, as $dist( i_1, i_2) = \sum_{ip\in [i_1, i_2]} w_{ip}$,
%and the distance between two networks, as the minimum distance between two IPs in different networks.
%Notice that this is not a distance function in its strict mathematical since, it can be positive or negative.

\subsubsection*{Collateral Damage and Benefit}
We define the collateral damage of a filtering solution as
$\sum_{p/l} \sum_{ip \in p/l \cap \mathcal{WL}} w_{ip} \cdot x_{p/l}$,
\ie the sum of the weights of the good addresses whose traffic is blocked.
We define the filtering benefit as
$\sum_{p/l} \sum_{ip \in p/l \cap \mathcal{BL}} w_{ip} \cdot x_{p/l}$,
\ie the sum of the weights of the bad addresses whose traffic is blocked.

\subsection{Rationale and Overview of Filtering Problems}

Given a set of bad and a set of good source addresses ($\mathcal{BL}$ and $\mathcal{WL}$),
a measure of their importance (the address weights $w$), and a resource budget ($F_{max}$ plus, possibly, other resources,
depending on the particular problem), the goal is to select which source prefixes to filter so as to
minimize the impact of bad traffic and can be accommodated with the given resource budget.
Different variations of the problem can be formulated, depending on the attack scenario and the victim network's policies and constraints: the network operator may want to block all bad addresses or tolerate to leave some unblocked;
the attack may be of low rate or a flooding attack; filters may be installed at one or several routers.

\begin{table}[t!]
\setlength{\tabcolsep}{5pt}		% tweak space between columns; default is 6
   \centering
   %\topcaption{Table captions are better up top} % requires the topcapt package
   \begin{tabular}{@{} |c|c| } % Column formatting, @{} suppresses leading/trailing space
	\hline
	$ip$ &  Generic IP address \\
	\hline
     $w_{ip}$ & Weight assigned to address $ip$\\
      \hline
      $\mathcal{BL}$ & Blacklist: a list of bad addresses\\
      \hline
      $N = |\mathcal{BL}| $ & Number of unique addresses in $\mathcal{BL}$\\
      \hline
      $\mathcal{WL} $ %% m=2^{32}-1
          & Whitelist: a set of ``good'' addresses \\
        \hline
    $p/l$ (or ``$p$'' for short) & Prefix $p$ of length $l$ bits (IP/mask notation) \\
    \hline
     $ip \in p/l$ &  IP  address that belongs to prefix $p/l$ \\
     %~  & (\ie $ip \in [p\cdot 2^{32-l}, (p+1)\cdot 2^{32-l}-1]$) \\
    \hline
    %$[l,r]$ & continuous address range from left ($l$) to right ($r$) \\
    %\hline
    %$\mathcal D$ &            set of source address prefixes as in Eq.(\ref{eqn: prefixes domain})\\
    % \hline
    % $[l,r] \in \mathcal D$ & range $[l,r]$ that is a valid prefix\\
    % ~ &                      (we only study these types of filters)\\
    % \hline
    %$x_{l,r}\in \{1,0\}$ & indicates  if filter $[l, r]$ is used or not \\
    %\hline
    $x_{p/l}\in \{1,0\}$ & Indicates if a filter blocks prefix $p/l$ or not \\
    \hline
    %  $ g_{l,r}=\sum_{ip\in[l,r] \cap G} w_{ip} $  & collateral damage from filter $[l, r]$\\
    % \hline
    %  $ b_{l,r}=\sum_{ip\in[l,r] \cap B} w_{ip}  $  & bad traffic filtered by filter $[l, r]$ \\
    % \hline
    $ g_{p/l}=\sum_{ip\in p/l \cap \mathcal{WL}} w_{ip} $  & Collateral damage from filtering prefix $p/l$\\
     \hline
     $ b_{p/l}=|\sum_{ip\in p/l \cap \mathcal{BL}} w_{ip}|  $  & Bad traffic blocked by filtering prefix $p/l$ \\
    % \hline
     % $F$ & Number of filters used \\
      \hline
      $F_{max}~(<<N)$ & Maximum number of available filters\\
     \hline
     $z_p(F)$        & Value of the optimal solution of subproblem  \\
      ~                     & considering only addresses in prefix $p$  \\
      			 & and up to $F$ filters \\
%     (or $z_p(F,C)$) & (and capacity $C$, in the case of FLOODING)\\
     \hline
     $X_p(F)$        & Set of filters used in optimal solution $z_p(F)$\\
    \hline
   \end{tabular}
   \caption{Summary of Notation and Terminology}
   \vspace{-10pt}
   \label{tab:booktabs}
\end{table}

At the core of each filtering problem lies the following optimization problem:
\begin{align}
\label{P-GENERAL} \min \sum_{p/l} \sum_{ip \in p/l} w_{ip}  x_{p/l} & \\
\label{P-cardinality} \text{s.t.~} \sum_{p/l} x_{p/l} & \leq F_{max}\\
%\label{P3-coverage}\sum_{[l,r] \ni i} x_{l,r}  & = 1\quad&&\forall ip\in \mathcal B && \\
\label{P-coverage} \sum_{p/l: ip \in p/l} x_{p/l}  & \leq 1\quad \forall ip\in \mathcal{BL}\\
\label{P-domain} x_{p/l} \in \{0,1\} &	\quad  \forall l=0,...,32, p=0,...,2^l  %\text{Allowed Domain}
\end{align}

Eq. (\ref{P-GENERAL}) expresses the objective to minimize the total cost of bad traffic,
which consists of two parts: the collateral damage (the terms with $w_{ip}>0$) and the cost of
leaving bad traffic unblocked (the terms with $w_{ip} < 0$).
We use notation $\sum_{p/l}$ to denote summation over all possible prefixes $p/l$: $l=0,...,32$, $p=0,...,2^l-1$.
Eq. (\ref{P-cardinality}) expresses the constraint on the number of filters.
Eq. (\ref{P-coverage}) states that overlapping filters are mutually exclusive, \ie each bad address can be blocked
at most once, otherwise filtering resources are wasted.
Eq. (\ref{P-domain}) lists the decision variables $x_{p/l}$ corresponding to all possible prefixes,
 and will be omitted from now on for brevity.

Eq. (\ref{P-GENERAL})--(\ref{P-domain}) provide the general framework for filter-selection optimization. Different
  filtering problems can be written as special cases, possibly with additional constraints.
As we discuss in Section \ref{sec:related}, these are all multi-dimensional knapsack problems \cite{KPbook},
which are, in general, NP-hard. The specifics of each problem affect dramatically the complexity, which can vary from linear to NP-hard.

In this paper, we formulate five practical filtering problems and develop optimal, yet computationally efficient algorithms to solve them.
Here, we summarize the rationale behind each problem and we outline our main results; the exact formulation and detailed solution
is provided in Section~\ref{sec:algorithms}.

{\bf {\em BLOCK-ALL}:} Suppose a network operator has a blacklist $\mathcal{BL}$ of size $N$, a whitelist $\mathcal{WL}$,
and a weight assigned to each address that indicates the amount of traffic originating from that address.
The total number of available filters is $F_{max}$.
The first practical goal the operator may have is to install a set of filters that block {\em all} bad traffic
so as to minimize the amount of good traffic that is blocked.
We design an optimal algorithm that solves this problem at the lowest achievable complexity (linearly increasing with $N$).

{\bf {\em BLOCK-SOME}:} A blacklist and a whitelist are given, as before, but the
operator is now willing to block {\em only some}, instead of all, bad traffic,
so as to decrease the amount of good traffic blocked at the expense of leaving some bad traffic unblocked.
The goal now is to block only those prefixes that have the highest impact and do not contain sources that generate a lot of traffic, so as to minimize the total cost in Eq. (\ref{P-GENERAL}).
We design an optimal, lowest-complexity (linearly increasing with $N$) algorithm for this problem, as well.

{\bf {\em TIME-VARYING BLOCK-ALL (SOME)}:}
%So far, we have considered two ``static'' versions of the filtering problem, where we have a fixed blacklist and whitelist. However,
Bad addresses may change over time~\cite{clustering}: new sources may send malicious traffic and conversely, previously active sources may disappear (\eg when their vulnerabilities are patched).
 %We  consider the dynamic version of the filtering problems, where the blacklist and whitelist change over time.
%, such that the operator is essentially given a sequence of blacklists,
%$\{ \mathcal{BL}_{T_0}, \mathcal{BL}_{T_1}, \dots, \mathcal{BL}_{T_{i}}, \dots \}$,
%and a sequence of whitelists, $\{ \mathcal{WL}_{T_0}, \mathcal{WL}_{T_1}, \dots, \mathcal{WL}_{T_{i}}, \dots \}$,
%at times $T_0<T_1<\dots<T_{i}<\dots$.
One way to solve the dynamic versions of BLOCK-ALL (SOME) is to run the algorithms we propose for the static versions for the blacklist/whitelist pair at each time slot.
However, given that subsequent blacklists typically exhibit significant overlap \cite{clustering}, it may be more efficient to exploit this temporal correlation
and incrementally update the filtering rules.
%Hence, the goal of $P_3$ is to construct a more efficient solution to BLOCK-ALL(SOME) for input $\mathcal{BL}_{T_{i}}$ and $\mathcal{WL}_{T_{i}}$, given the solution to BLOCK-ALL(SOME) for input $ \mathcal{BL}_{T_{i-1}}$ and $\mathcal{WL}_{T_{i-1}}$.
We show that is it possible to update the optimal solution, as new IPs are inserted in or removed from the blacklist, in $\log N$ time.

{\bf  {\em FLOODING}:} In a  flooding attack, such as the one shown in Fig.~1,
a large number of compromised hosts send traffic to the victim and exhaust the victim's access bandwidth.
In this case, our framework can be used to select the filtering rules that minimize the amount of good traffic that is blocked
while meeting the access bandwidth constraint -- in particular, the total bandwidth consumed by the unblocked traffic should not exceed
the bandwidth of the flooded link, \eg link G-V in Fig.~1.
We prove that this problem is NP-hard and we design a pseudo-polynomial algorithm that solves it optimally,
with complexity that grows linearly with the blacklist and whitelist size, \ie $|\mathcal{BL}|+|\mathcal{WL}|$.

%{\bf [$P_4$] {\em FILTER-DYNAMIC}:}
%If the attack profile changes over time, the filtering rules must be updated accordingly.
%instead of a single blacklist, we may be given a set of blacklists over time
%$BL = \{ BL_{t_0}, BL_{t_1},\cdots\}$  and a number of filters, $F_{max}$. The dynamic versions of the problems
%should find filtering rules $\{S_{t_0}, S_{t_1},\cdots\}$ at every time $t_{ip}$, such that
%$S_{t_{ip}}$ solves problem $P_1$/$P_2$/$P_4$ when the input list is $BL_{t_{ip}}$.
%The trivial solution to the dynamic problem would be to solve the corresponding static problem ( $P_1$,$P_2$ or $P_4$)
%every time from scratch. However, the temporal correlation in the behavior of bad sources allows
%for greedy algorithms that incrementally update the solution.

%%P6
{\bf {\em DIST-FLOODING}:}
All the above problems aim at installing filters at a single router.
However, a network operator may use the filtering resources collaboratively across {\em several routers} to better defend against an attack.
Distributed filtering
may also be enabled by the cooperation across several ISPs against a common enemy.
The question in both cases is not only which prefixes to block, but also at which router to install the filters.
We study the practical problem of distributed filtering against a flooding attack.
We prove that the problem can be decomposed into several FLOODING problems, which can be solved in a distributed way.

\section{\label{sec:algorithms}Filtering Problems and Algorithms}

In this section, we provide the detailed formulation of each problem and we present the algorithm that solves it. We start by defining the data structure that we use to represent the problem and to develop our algorithms.

\subsection{A Data Structure for Representing Filtering Solutions}

\begin{definition}[LCP Tree]
Given a set of addresses $\mathcal A$, we define {\em the Longest Common Prefix tree} of $\mathcal A$, denoted by LCP-tree($\mathcal A$),
as the binary tree with the following properties:
(i) each leaf represents a different address in $\mathcal{A}$ and there is a leaf for each address in $\mathcal{A}$;
(ii) each intermediate (non-leaf) node represents the longest common prefix between the prefixes represented by its two children.
\end{definition}

The LCP-tree($\mathcal{A}$) can be constructed from the complete binary tree (with root  leaves at level 32 corresponding to all addresses $[0,....,2^{32}-1]$, and intermediate nodes at level $i=1,..32$ corresponding to all prefixes of length $i$) by removing the branches that do not have addresses in $\mathcal{A}$, and then by removing nodes with a single child.
It is a variation of the binary (or unibit) trie~\cite{varghese}, but does not have nodes with a single child.
%We do not claim novelty in this data structure, but we describe it in detail because we use it extensively in our algorithms.
The LCP-tree($\mathcal{A}$) offers an intuitive way to represent sets of prefixes that can block the addresses in set $\mathcal{A}$: each node in the LCP tree represents a prefix that can be blocked, hence we can represent a filtering solution as the pruned version of the LCP tree, whose leaves are all and only the blocked prefixes.

\begin{example}
%An example is shown and discussed in Fig.~\ref{fig:P3-proof}.
For instance, consider the LCP tree depicted in Figure~\ref{fig:P3-proof}, whose leaves correspond to bad addresses that we want to block. One (expensive) solution is to use one filter to block each bad address; thus the LCP tree is not pruned and its leaves correspond to the filters. Another feasible solution is to use three filters and block traffic from prefixes $0/1$, $8/2$, and $12/4$; this can be represented
by the pruned version of the LCP tree that includes the
aforementioned prefixes as leaves. Yet another (rather radical) solution is to filter a single prefix ($0/0$) to  block all traffic; this can be represented by the pruned version of the LCP tree that includes only its root.
\end{example}.

\begin{figure}
\centering
\includegraphics[width=3.5in, angle=0]{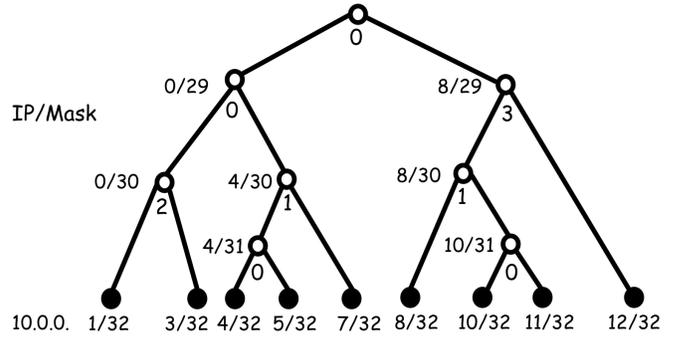}
\caption{\label{fig:P3-proof} {\small {\bf Example of LCP-tree($\mathcal{BL}$).}
Consider a blacklist consisting of the following 9 bad addresses:
$\mathcal{BL}=\{10.0.0.1,10.0.0.3,10.0.0.4,10.0.0.5,10.0.0.7,10.0.0.8,10.0.0.10$, $10.0.0.11,10.0.0.12\}$. All remaining addresses are considered good.
Each leaf represents one address in the $\mathcal{BL}$.
Each intermediate node represents the longest common prefix $p$ covering all bad addresses in that subtree.
%Moreover, each intermediate node that represents prefix $p$ is associated with a cost equal to the additional collateral damage caused when we filter p instead of filtering each of its children.
%For ease of illustration, we set collateral damage as equal to the number of good addresses that get blocked.
At each intermediate node $p$, we also show the collateral damage (\ie number of good addresses blocked) when we filter prefix $p$ instead of filtering each of its children.
{\em E.g.,} if we use two filters to block bad addresses 10.0.0.5/32 and 10.0.0.7/32 the collateral damage is 0; if, instead, we use one filter to block prefix 10.0.0.4/30, we also block good address 10.0.0.6/4, \ie we cause collateral damage 1.}}
\end{figure}

%\begin{proof}
%A feasible solution of BLOCK-ALL select at most $F_{max}$ non-overlapping prefixes
%such that all bad IPs are included in exactly one of those prefixes. In the LCP tree this
%can be represented as a subtree, where the leaves correspond to the selected filters
%\end{proof}

%There is a connection between the LCP-tree and the filter optimization problem(s). A filtering problem becomes equivalent to choosing up to
%$F_{max}$ nodes in the tree (intermediate or leaves) to block the corresponding prefixes, so as to minimize the total cost and to meet the constraints of the problem. In other words, the LCP tree provides a characterization of all and only the feasible solutions of BLOCK-ALL.
%\textcolor{blue}{[Fabio: will you state a proposition?]}

\subsubsection*{Complexity}
%Constructing the LCP-tree from the $N$ malicious addresses in set $\mathcal S$
%(Note that an LCP-tree can be constructed from the binary tree of all prefixes, by removing the branches that do not %have malicious IP in its downstream nodes, and then removing nodes with a single child node).
%requires \textcolor{blue}{on average} $O(N\log N)$ operations since we can first
%sort the blacklist ($O(N\log N)$) and then compute the longest common prefix of length $p=31,\dots,1$ (if any) only between consecutive IPs ($\big\lceil \sum_{p=31}^{1}\frac{N}{2^p} \big\rceil \leq 31*\frac{N}{2} = O(N)$).
%We keep the list of longest common prefixes $\mathcal L$, which has size (at most) $N$.
Given a list of addresses $\mathcal{A}$, we can build  the LCP-tree($\mathcal{A}$) by performing $|\mathcal{A}|$ insertions in a Patricia trie~\cite{varghese}.
To insert a string of $m$ bits, we need at most $m$ comparisons. Thus, the worst-case complexity is $O(m|\mathcal{A}|)$, where $m=32$ (bits) is the  length of a 32-bit IPv4 address.

\subsection{BLOCK-ALL}
\label{sec:block-all}

\subsubsection*{Problem Statement}
Given: a blacklist $\mathcal{BL}$, a whitelist $\mathcal{WL}$, and the number of available filters $F_{max}$;
select filters that block {\em all} bad traffic and minimize collateral damage.

\subsubsection*{Formulation}
We formulate this problem by making two adjustments to the general framework of Eq. (\ref{P-GENERAL})--(\ref{P-domain}).
First, Eq.~(\ref{P-GENERAL}) becomes Eq.~(\ref{P1-OF}) below, which expresses the goal to minimize the collateral damage. %(as opposed to minimizing the total cost of the attack).
Second, Eq.~(\ref{P-coverage}) becomes Eq. (\ref{P1-coverage}) below,
which enforces the constraint that every bad address should be blocked by exactly one filter,
as opposed to at most one filter in Eq.(\ref{P-coverage}).
\begin{align}
\label{P1-OF} \min \sum_{p/l} g_{p/l} x_{p/l} \\
\label{P1-cardinality} \text{s.t.}  \quad \sum_{p/l} x_{p/l} &\leq F_{max} \\
\label{P1-coverage} \sum_{p/l: ip \in p/l} x_{p/l}   &= 1 & \quad \forall ip\in \mathcal{BL}
%\label{P1-domain} &x_{p/l} \in \{0,1\}  \quad & \forall l=0,..32, p=0,..2^l
\end{align}

\subsubsection*{Characterizing an Optimal Solution}
Our algorithm starts from LCP-tree($\mathcal{BL}$) and outputs a pruned version of that LCP tree.
Hence, we start by proving that an optimal solution to BLOCK-ALL can indeed be represented as a pruned version of that LCP tree.

\begin{proposition}\label{prop:all}
An optimal solution to BLOCK-ALL can be represented as a pruned
subtree of LCP-tree($\mathcal BL$) with the same root as LCP-tree($\mathcal{BL}$), up to $F_{max}$ leaves,
and each non-leaf node having exactly two children.
\end{proposition}

\begin{proof}
We prove that, for each feasible solution to BLOCK-ALL $S$, there exists another feasible solution $S'$
that (i) can be represented as a pruned subtree of LCP-tree($\mathcal{BL}$) as described
in the proposition and (ii) whose collateral damage is smaller or equal to $S$'s.
This is sufficient to prove the proposition, since an optimal solution is also a feasible one.

Any filtering solution can be represented as a pruned subtree of
the full binary tree of all IP addresses (LCP-tree$\left( \{0, 1, ..., 2^{32}-1 \} \right )$)
with the same root and leaves corresponding to the filtered prefixes.
$S$ is a feasible solution to BLOCK-ALL, therefore 
$S$ uses up to $F_{max}$ filters, \ie its tree has up to $F_{max}$ leaves. Indeed, if this was not the case,
Eq.~(\ref{P1-cardinality}) would be violated and $S$ would not be a feasible solution.

Let us assume that the tree  representing $S$ includes a prefix $\tilde p$ that is {\em not} in LCP-tree($\mathcal{BL}$). There are three possible cases:
\begin{enumerate}
\item $\tilde p$ includes no bad addresses.
In this case, we can simply remove $\tilde p$ from $S$'s tree (\ie, unblock $\tilde p$).
\item Only one of $\tilde p$'s children includes bad addresses.
In this case, we can replace $\tilde p$ with the child node.
\item Both of $\tilde p$'s children contain bad addresses. In this case, $\tilde p$ is already the longest common prefix of all bad addresses in $\mathcal{BL}\cap \tilde p$, thus already on the LCP-tree($\mathcal{BL}$).
\end{enumerate}
Clearly, each of these operations transforms feasible solution $S$, which is assumed not to be on LCP-tree($\mathcal{BL}$), into another feasible solution $S'$ with smaller or equal collateral damage but on the LCP-Tree($\mathcal{BL}$). We can repeat this process for all prefixes that are in $S$'s tree but not in LCP-tree($\mathcal{BL}$), until we create a feasible solution $S'$ that includes only prefixes from LCP-tree($\mathcal{BL}$) and has smaller or equal collateral damage.
%That is, for each feasible solution $S$, there exists a better feasible solution $S'$, which can be represented as a pruned subtree of LCP-tree($\mathcal{BL}$) with the same root as LCP-tree($\mathcal{BL}$) and up to $F_{max}$ leaves.

The only element missing to prove the proposition is to show that, in the pruned LCP subtree that represents $S'$,
each non-leaf node has exactly two children.
We show this by contradiction:
Suppose there exists a non-leaf node in our pruned LCP subtree that has exactly one child.
%Given that each non-leaf node in the LCP tree has exactly two children, a node with a single child node 
This can only result from pruning out one child of a node in the LCP tree.
This means that all the bad addresses (leaves) in the subtree of this child node remain unfiltered, which violates Eq.~(\ref{P1-coverage}); but this is a contradiction because  $S'$ is a feasible solution.
\end{proof}

%
%Moreover, note that the LCP-trees represents all and only the prefixes that need to be considered in the optimal filter allocation. We cannot remove any node from it without affecting the optimality of the solution, and adding another prefix does not bring any benefit but increases the computation and storage.
%More details can be found in the technical report .

%\textcolor{blue}{For example in Fig.2.... }

\subsubsection*{Algorithm} Algorithm \ref{Alg:P1}, which solves BLOCK-ALL, consists of two steps. First, we build the LCP tree from the input
blacklist $\mathcal{BL}$. Second, in a bottom-up fashion, we compute $z_p(F )\ \forall p,F$, \ie, the minimum collateral damage needed to block all
bad addresses in the subtree of prefix $p$ using at most $F$ filters.
Following a dynamic programming (DP) approach,  we can find the optimal allocation of filters in the subtree rooted at prefix $p$,
by finding a value $n$ and assigning $F-n$ filters to the left subtree and $n$ to the right subtree, so as to minimize collateral damage.
The fact that BLOCK-ALL needs to filter all bad addresses (leaves in the LCP tree) implies that at least one filter must be assigned to the left
and right subtree, \ie $n=1,2,...,F-1$.
%Using this observation, we can design an efficient bottom-up scheme which recursively computes the optimal solution.
%Starting from the leaves, we can use this observation to efficiently solve problem ???
%Thus, we can solve problem ??? using a bottom-top approach.
In other words, for every pair of sibling nodes, $s_l$ (left) and $s_r$ (right), with common parent node $p$,  the following recursive equation holds:
\begin{align}
\label{eq:block-all} z_p(F) &=  \min_{n=1,..., F-1} \Big\{ z_{s_l}(F-n) +  z_{s_r}(n) \Big\}, ~F>1
\end{align}
with boundary conditions for leaf and intermediate nodes:
\begin{align}
%\label{cond1:block-all} z_{leaf}(F) & = 0	 \quad \forall F\geq 1 &\\
%\label{cond2:block-all} z_p(1) & = \textcolor{blue}{g_{p}} \quad \forall p
\label{conds:block-all}
z_{leaf}(F) & = 0   \  \ \forall F\geq 1  \\
 z_p(1)  & = g_{p} \  \ \forall p
\end{align}
Once we compute $z_p(F)$ for all prefixes in the LCP tree, we simply read the value of the optimal solution, $z_{root}(F_{max})$.
We also use auxiliary variables $X_p(F)$ to keep track of the set of prefixes used in the optimal solution.
In lines 4 and 10 of Algorithm 1,  $X_p(F)$ is initialized to the single prefix used.
In line 12, after computing the new cost, the corresponding set of prefixes is updated: $X_p(F) = X_{s_l}(F-n) \cup X_{s_r}(n)$.

\begin{algorithm}[t!]
 \caption{\label{Alg:P1} Algorithm forsolving  BLOCK-ALL}
\begin{algorithmic}[1]
\begin{footnotesize}
\STATE build LCP-tree($\mathcal{BL}$)
\FOR{ all leaf nodes $leaf$}
	\STATE{ $z_{leaf}(F) = 0 \ \forall F\in[1, F_{max}]$ }
	\STATE{$X_{leaf}(F) = \{leaf\} \ \forall  F\in[1, F_{max}]$ }
\ENDFOR
\STATE level = level(leaf)-1
\WHILE{$level \geq level(root)$} \label{line:while}
    \FOR{all node  $p$ such that $level(p)==level$ }
	\STATE $z_p(1) = g_p$
	\STATE $X_p(1) = \{p\}$
	\STATE $z_p(F)=\min_{n=1,..F-1} \Big\{ z_{s_l}(F-n)+ z_{s_r}(n) \Big\} \forall F\in[2,F_{max}]$
   	\STATE $X_p(F) = X_{s_l}(F-n) \cup X_{s_r}(n) \forall F\in[2, F_{max}]$
    \ENDFOR
    \STATE level = level - 1
\ENDWHILE
\RETURN $z_{root}(F_{max})$, $X_{root}(F_{max})$
\end{footnotesize}
\end{algorithmic}
\end{algorithm}

\begin{theorem}
Algorithm~1 computes the optimal solution of problem BLOCK-ALL: the prefixes that are contained in set $X_p(F)$ are the optimal
$x_{p/l}=1$ for Eq. (\ref{P1-OF})--(\ref{P1-coverage}).
\end{theorem}
\begin{proof}
Recall that, $z_{root}(F_{max})$ denotes the value of the optimal solution of BLOCK-ALL with $F_{max}$ filters
(\ie, the minimum collateral damage), while $X_{root}(F_{max})$ denotes the  set of filters selected in the optimal solution.
Let $s_l$ and $s_r$ denote the two children nodes (prefixes) of $root$ in the LCP-tree($\mathcal{BL}$).
Finding the optimal allocation of $F_{max}>1$ filters to block all addresses contained in $root$ (possibly all IP space),
is equivalent to finding the optimal allocation of $x\geq 1$ filters to block all addresses in $s_l$,
and $y\geq 1$ prefixes for bad addresses in $s_r$, such that $x+y = F_{max}$.
This is because prefixes $s_l$, and $s_r$ jointly contain {\em all} bad addresses.
%Thus, filtering prefixes not included neither in  $s_l$, nor in  $s_r$ (if any), can either increase or leave unchanged the amount of collateral damage we want to minimize.
Moreover, each of $s_l$ and $s_r$ contains at least one bad address. Thus, at least one filter must be assigned to each of them.
If $F_{max}=1$, \ie, there is only one filter available, the only feasible solution is to select $root$ as the prefix to filter out.
The same argument recursively applies to descendant nodes, until either we reach a leaf node, or we have only one filter available.
In these cases, the problem is trivially solved by Eq. (\ref{conds:block-all}).
\end{proof}

\subsubsection*{Complexity} The LCP-tree is a binary tree with $|\mathcal{BL}|$ leaves; therefore, it has $O(|\mathcal{BL}|)$ intermediate nodes (prefixes). Computing Eq. (\ref{eq:block-all}) for every node $p$ and for every value $F\in[1, F_{max}-1]$ involves solving $O(|\mathcal{BL}|F_{max})$ sub-problems,
one for every pair ($p$, $F$) with complexity $O(F_{max})$.
$z_p(F)$ in Eq. (\ref{eq:block-all}) requires only the optimal solution at the
sibling nodes, $z(s_l, F-n),  z(s_r, n)$. Thus, proceeding from the leaves to the root, we can compute the optimal solution in
$O(|\mathcal{BL}|F_{max}^2)$. In practice, the complexity is even lower, since we do not need to compute $z_p(F)$ for all values
$F \leq F_{max}$, but only for $F\leq \min\{|leaves(p)|, F_{max}\}$, where $|leaves(p)|$ is the number of the leaves in prefix $p$ in the LCP tree.
Moreover, we only need to compute entries $z_p(F)$ for every prefix $p$, s.t. we cover all addresses in $\mathcal{BL}\cap p$, which may require $F \leq F_{max}$ for long prefixes in the LCP-tree.

Finally, we observe that the asymptotic complexity is $O(|\mathcal{BL}|)$,
since  $F_{max}<<N=|\mathcal{BL}|$ and $F_{max}$ does not depend on $|\mathcal{BL}|$ but only on the TCAM size.
Thus, the time complexity increases linearly with the number of bad addresses $|\mathcal{BL}|$.
This is within a constant factor of the lowest achievable complexity, since we need to read all $|\mathcal{BL}|$ bad addresses at least once.

\subsection{\label{sec:block-some}BLOCK-SOME}

\subsubsection*{Problem Statement}
Given a blacklist $\mathcal{BL}$, a whitelist $\mathcal{WL}$, and the number of available filters $F_{max}$,
the goal is to select filters so as to minimize the total cost of the attack.

\subsubsection*{Formulation}
This is precisely the problem described by Eq. (\ref{P-GENERAL})--(\ref{P-domain}),
but put slightly rephrased to better compare it with BLOCK-ALL.
There are two differences from BLOCK-ALL. First, the goal is to minimize the total cost of the attack, which involves
both collateral damage $g_{p/l}$ and the filtering benefit $b_{p/l}$, which is expressed by Eq. (\ref{P2-OF}).
Second, Eq. (\ref{P2-coverage}) states that every bad address must be filtered by {\em at most} one prefix, which means that it may or may not be filtered.
\begin{align}
\label{P2-OF} \min & \sum_{p/l} \Big(g_{p/l} - b_{p/l}\Big) x_{p/l}&\\
\label{P2-cardinality} \text{s.t.} &\sum_{p/l} x_{p/l} \leq F_{max} &\\
\label{P2-coverage}  &\sum_{p/l: ip\in p/l} x_{p/l} \leq 1\quad &\forall ip\in \mathcal{BL} &&
%\label{P2-domain} &x_{p/l}  \in \{0,1\} 	\quad & \forall l=0,..32, p=0,..2^l &&
\end{align}

%\begin{algorithm}[t!]
% \caption{\label{Alg:P2} {\em \textcolor{blue}{DP} Algorithm for BLOCK-SOME}}
%\begin{algorithmic}[1]
%\begin{footnotesize}
%\STATE build LCP-tree
%\STATE level = level(leaf)
%\WHILE{$level \geq level(root)$} \label{line:while}
%    \FOR{all node  $s$ such that level(s)==level }
%    	\STATE $\tilde F = \min\{F_{max}, \#descendants\}$
%        \STATE solve Problem Eq.(\ref{eq:block-some})-(\ref{cond3:block-some}) for \textcolor{blue}{$F=0,...\tilde F$} %$\forall F\leq \tilde F$
%    \ENDFOR
%    \STATE level = level - 1
%\ENDWHILE
%\RETURN $z(root, F_{max})$
%\end{footnotesize}
%\end{algorithmic}
%\end{algorithm}

\subsubsection*{Characterizing an Optimal Solution}
As with BLOCK-ALL, our algorithm starts from LCP-tree($\mathcal{BL}$) and outputs a pruned version of that LCP tree.
The only difference is that some bad addresses may now remain unfiltered.
In the pruned LCP subtree that represents our solution, this means that
there may exist intermediate (non-leaf) nodes with a single child.

\begin{proposition}\label{prop:some}
An optimal solution to BLOCK-SOME can be represented as a pruned subtree of
LCP-tree($\mathcal BL$) with: the same root as LCP-tree($\mathcal BL$) and up to $F_{max}$ leaves.
\end{proposition}

\begin{proof}  
In Proposition~\ref{prop:all}, we proved that any solution of Eq. (\ref{P1-OF})--(\ref{P1-cardinality}) can be reduced to a (pruned)
subtree of the LCP tree with at most $F_{max}$ leaves. Moreover, the constraint expressed by Eq.~(\ref{P2-coverage}), 
which imposes the use of non-overlapping prefixes, is automatically imposed considering the {\em leaves} of the pruned subtree as the selected filter. 
This proves that any feasible solution of BLOCK-SOME can be represented as a pruned subtree of the LCP tree with at most $F_{max}$ leaves. 
And thus, so can an optimal solution.
\end{proof}

%\begin{proposition} \label{prop:some}
%Given $\mathcal{BL}$ and $F_{max}$, there exists an optimal solution of BLOCK-SOME that can be represented as a pruned %subtree of LCP($\mathcal BL$) with the same root and $F_{max}$ leaves.
%\end{proposition}

%{\em Proof Sketch.} Similar to Prop.\ref{prop:all} \cite. The difference from BLOCK-ALL is that, because some bad %IPs can remain unfiltered, the pruned LCP tree corresponding to a feasible solution can have nodes with a single descendant. \hfill $\blacksquare$

%\textcolor{blue}{For example in Fig.2...}

\subsubsection*{Algorithm} The algorithm that solves BLOCK-SOME is similar to Algorithm \ref{Alg:P1} in that it
relies on the LCP tree and a dynamic programming (DP) approach. The main difference is that not all bad addresses need to be filtered, hence, at each step,
we can assign $n=0$ filters to the left and/or right subtree.
More specifically, whereas in line (11) of Algorithm~\ref{Alg:P1} we had $n=1,...,F-1$, now we have $n=0,1,...,F$.
We can recursively compute the optimal solution as before:
%There are now two possible operations on the LCP-tree: either to {\em merge} two leaves as before, or to {\em remove}
%a leaf node (the one with the maximum weight). The cost of a removal operation is equal to cost of the removed leaf. In every iteration, Alg.\ref{Alg:P2} still reduces the number of filters by one: it chooses among the two %possible operations, whichever brings the smallest increase in the objective function.
%\vspace{-0.2cm}
\begin{align}
\label{dp:block-some}z_p(F) &= \min_{n=0,..., F} \Big\{ z_{s_l}(F-n) + z_{s_r}(n) \Big\}
\end{align}
with boundary  conditions
\begin{align}
\label{cond1:block-some} &z_p(0) = 0 \quad \forall~p \\
\label{cond2:block-some} &z_p(1)  = \min \Big\{  g_p-b_p,  \min_{n=0,1} \Big\{ z_{s_l}(1-n)+z_{s_r}(n) \Big\} \Big\}  \\
\label{cond3:block-some} &z_{leaf}(F) = -b_{leaf}	 \quad \forall F\geq 1
\end{align}
where  $p$ is an intermediate node (prefix) and $leaf$ is a leaf node in the LCP-tree.

\subsubsection*{Complexity} The analysis of Algorithm~\ref{Alg:P1} applies to this algorithm as well.
The complexity is the same, \ie, linearly increasing with $|\mathcal{BL}|$.

\subsubsection*{BLOCK-ALL vs. BLOCK-SOME}
There is an interesting connection between the two problems.
The latter can be regarded as an automatic way to select the best subset from $\mathcal{BL}$
and run BLOCK-ALL only on that subset.
If the absolute value of weights of bad addresses are significantly larger than the weights of the good addresses,
then BLOCK-SOME degenerates to BLOCK-ALL.

\begin{figure}[t!]
\begin{center}
\includegraphics[width=3.3in ]{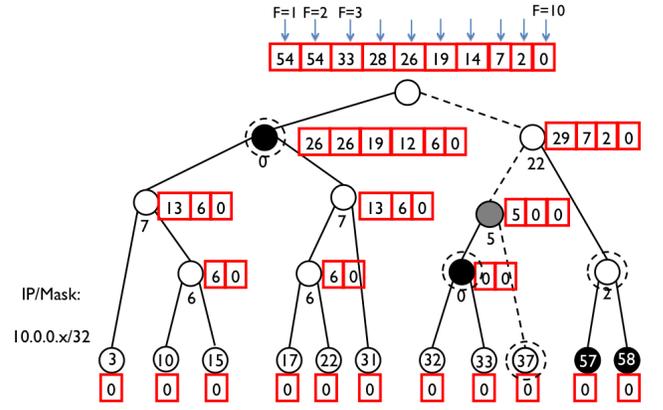}
\caption{\label{fig:dynamic}
{\bf Example of (i) BLOCK-ALL and (ii) TIME-VARYING BLOCK-ALL.}
Consider a blacklist of 10 bad IP addresses $\mathcal{BL} = \{10.0.0.3,10.0.0.10,10.0.0.15,10.0.0.17,10.0.0.22,10.0.0.31,10.0.0.32,$ $10.0.0.33,10.0.0.57,10.0.0.58\}$
% For ease of illustration, we consider a $6$-bit (instead of 32-bit) IP addresses and show them in decimal (instead of $a.b.c.d$) notation. Consider  blacklist $\mathcal{BL} = \{3,10,15,17,22,31,32,33,57,58\}$.
The table next to each node $p$ shows the minimum cost $z_p(F)$ computed by the DP algorithm for BLOCK-ALL for $F=1,...\textit{number of leaves in subtree}$. The optimal solution to BLOCK-ALL consists of the 4 prefixes highlighted in black. When a new address, \eg $10.0.0.37$, is added to the blacklist, a leaf node is added to the tree and TIME-VARYING needs to update all and only the predecessor nodes in LCP-tree($\mathcal{BL}$), indicated by the dashed lines, according to Eq. (\ref{eq:block-all}). Moreover, a new node is created to denote the longest common prefix between $10.0.0.37$ and $10.0.0.32$ (or $10.0.0.33$). Note that all other nodes corresponding to the longest common prefixes between $10.0.0.37$ and other addresses in $\mathcal{BL}$ are already in the LCP tree. The new optimal solution consists of the 4 prefixes indicated by the dashed circles.}
\end{center}
\end{figure}

\subsection{\label{sec:time-varying}TIME-VARYING BLOCK-ALL(SOME)}

 We now consider that the blacklist and whitelist change over time and we seek to incrementally update the filtered prefixes that are affected by the change. More precisely, consider a sequence of blacklists, $\{ \mathcal{BL}_{\tau_0}, \mathcal{BL}_{\tau_1}, \ldots \}$ and of whitelists, $\{ \mathcal{WL}_{\tau_0}, \mathcal{WL}_{\tau_1}, \ldots \}$ at times $ \tau_0, \tau_1, \ldots $, respectively.
%One way to solve BLOCK-ALL for each (blacklist, whitelist pair) is to run the algorithm from
%Section~\ref{sec:block-all} %(Section~\ref{sec:block-some})
%for each pair. However, in practice there is temporal correlation \cite{clustering, uncleanness}, which can be
%exploited by incrementally updating only the prefixes that are affected by the change.

\subsubsection*{Problem Statement}
Given:
(i) a blacklist and whitelist, $\mathcal{BL}_{\tau_{i-1}}$ and $\mathcal{WL}_{\tau_{i-1}}$, (ii) the number of available filters $F_{max}$,
(iii) the corresponding solution to BLOCK-ALL(SOME), denoted by $\mathcal{S}_{\tau_{i-1}}$, and
(iv) another blacklist and whitelist, $\mathcal{BL}_{\tau_{i}}$ and $\mathcal{WL}_{\tau_i}$;
obtain the solution to BLOCK-ALL(SOME) for the second blacklist/whitelist, denoted by $\mathcal{S}_{\tau_{i}}$.

\subsubsection*{Algorithm} Consider, for the moment, that the whitelist remains the same and focus on the changes in the blacklist.

{\em (i) Addition.} First, consider that the two blacklists differ only in a single new bad address, which does not appear in $\mathcal{BL}_{\tau_{i-1}}$, but appears in $\mathcal{BL}_{\tau_i}$. There are two cases, depending on whether the new bad address belongs to a prefix that is already filtered in $\mathcal{S}_{\tau_{i-1}}$. If it is, no further action is needed, and $\mathcal{S}_{\tau_{i}} = \mathcal{S}_{\tau_{i-1}}$. Otherwise, we modify the LCP tree that represents $\mathcal{S}_{\tau_{i-1}}$ to also include the new bad address, as illustrated in Fig.~\ref{fig:dynamic}.
The key point is that we only need to add one new intermediate node to the LCP tree (the gray node in Fig.~\ref{fig:dynamic}), corresponding to the longest common prefix between the new bad address and its closest bad address that is already in the LCP tree.
 The optimal allocation of $F$ filters to the subtree rooted at prefix $p$ depends only on how these $F$
filters are allocated to the children of $p$. Hence, when we add a new node to the LCP tree, we need to recompute the optimal filter allocation (\ie recompute $z_p(F)$ and $X_p(F) \ \forall F$, according to Eq.~(\ref{eq:block-all}))
for all and only the predecessors of the new node, all the way up to the root node.

{\em (ii) Deletion.} Then assume that two blacklists differ in one deleted bad address,
which appears in $\mathcal{BL}_{\tau_{i-1}}$ but not in $\mathcal{BL}_{\tau_i}$.
In this case, we modify the LCP tree that represents $\mathcal{S}_{\tau_{i-1}}$ to
we remove the leaf node that corresponds to that address as well as its parent node (since that node does not
have two children any more), and we recompute the optimal filter allocation for all and only the node's predecessors.

{\em (iii) Adjustment.} Finally, suppose that the two blacklists differ in one address, which appears in both blacklists but with different weights;
or, that the two blacklists are the same, while the two whitelists differ in one address (it either appears
in one of the two whitelists or it appears in both whitelists but with different weights).
In all of these cases, we do not need to add or remove any nodes from the LCP tree, but we do need to adjust the collateral
damage or filtering benefit associated with one node, hence recompute the optimal filter allocation for all and only that node's predecessors.

{\em (iv) Multiple addresses.} If the two successive time instances differ in multiple addresses, we repeat the procedures described above as needed, \ie we perform one node addition for each new bad address, one deletion for each removed bad address, and up to one adjustment for each other difference.

\subsubsection*{Complexity}
Since the LCP tree is a complete binary tree, any leaf node has at most $\log(|\mathcal{BL}|)$ predecessors,
so, inserting a new bad address (or removing one) requires $O(\log(|\mathcal{BL}|)F_{max}^2)$ operations.
Hence, deriving $\mathcal{S}_{\tau_{i}}$ from $\mathcal{S}_{\tau_{i-1}}$ as described above is asymptotically better than computing it from scratch using Algorithm~\ref{Alg:P1}.
 %or its variant from Section~\ref{sec:block-some}),
 if and only if the number of different addresses between the two time instances is less than $\frac{|\mathcal{BL}|}{\log |\mathcal{BL}|}$.

%%%%%%%%%%%%%%%%%%%%%%%%%%%%%%%%%%%%%%%%%%%%%%%%%%%%%%%
%%%%%%%%%%%%%%%%%%%%%%%%%%%%%%%%%%%%%%%%%%%%%%%%%%%%%%%
%%%%%%%%%%%%%%%%%%%%%%%%%%%%%%%%%%%%%%%%%%%%%%%%%%%%%%%
\subsection{\label{sec:flooding}FLOODING}

\subsubsection*{Problem Statement}
Given: (i) a blacklist $\mathcal{BL}$ and a whitelist $\mathcal{WL}$, where the absolute weight of each bad and good address  is equal to the amount
of traffic it generates, (ii) the number of available filters $F_{max}$, and (iii) a constraint on the victim's link capacity (bandwidth) $C$;
select filters so as to minimize collateral damage and make the total traffic fit within the victim's link capacity.

\subsubsection*{Formulation}
To formulate this problem, we need to make two adjustments to the general framework of Eq.~(\ref{P-GENERAL})--(\ref{P-domain}).
First, Eq.~(\ref{P-GENERAL}) becomes Eq.~(\ref{P3-OF}), which expresses the goal to minimize collateral damage.
Second, we add a new constraint Eq.~(\ref{P3-capacity}), which specifies that the total traffic that remains unblocked after filtering
(which is the total traffic, $T_0 = \sum_{ip \in \mathcal{BL} \cup \mathcal{WL}}w_{ip}$, minus the traffic
that gets blocked, $\sum_{p/l}\Big( g_{p/l} + b_{p/l}\Big) x_{p/l}$ should fit within the link capacity $C$,
so as to avoid congestion and packet loss.
\begin{align}
\label{P3-OF} \hspace{2.5cm}\min \sum_{p/l} g_{p/l} x_{p/l}&&\\
\label{P3-cardinality} \text{s.t.} \sum_{p/l}x_{p/l}  \leq F_{max} &&\\
\label{P3-capacity}   T_0 - \sum_{p/l}\Big( g_{p/l} + b_{p/l}\Big)  x_{p/l}  \leq C && \\
\label{P3-coverage}  \sum_{p/l: ip \in p/l} x_{p/l}   \leq 1\quad \forall ip \in \mathcal{BL} &&
\end{align}

\subsubsection*{Characterizing an Optimal Solution}
We represent the optimal solution as a pruned subtree of an LCP-tree.
However, we start with the full binary tree of all bad and good  addresses LCP-tree($\mathcal{BL} \cup \mathcal{WL}$). Moreover, to handle the constraint in Eq. (\ref{P3-capacity}), each node corresponding to
prefix $p$ is assigned an additional cost, $T_p$, indicating the total amount of traffic sent by $p$, $T_p = g_p + b_p$.

\begin{proposition}\label{prop:flooding}
An optimal solution of FLOODING can be represented as the leaves of a pruned subtree of LCP-tree($\mathcal{BL} \cup \mathcal{WL}$), with the same root, up to $F_{max}$ leaves, and total cost of the leaves $\geq T_0-C$.
\end{proposition}

\begin{proof}
Similarly to Proposition~\ref{prop:all},
we prove that for every feasible solution to FLOODING $S$, there exists another feasible solution $S'$,
which (i) can be represented as a pruned subtree of LCP-tree($\mathcal{BL} \cup \mathcal{WL}$) as described
in the proposition and (ii) whose collateral damage is smaller or equal to $S$'s.
This is sufficient to prove the proposition, since an optimal solution is also a feasible one.

Any filtering solution can be represented as a pruned subtree of
LCP-tree$\left( \{0, 1,..., 2^{32}-1 \} \right )$
with the same root and leaves corresponding to the filtered prefixes.
$S$ is a feasible solution to FLOODING, therefore:
$S$'s tree has up to $F_{max}$ leaves, otherwise Eq.~(\ref{P3-cardinality}) would be violated;
and the total cost of $S$'s leaves is $\geq T_0-C$, otherwise Eq.~(\ref{P3-capacity}) would be violated.

Suppose that  $S$ includes a prefix $\tilde p$ that is {\em not} in LCP-tree($\mathcal{BL} \cup \mathcal{WL}$). We can construct a better feasible solution $S'$, which
can be represented as a pruned subtree of LCP-tree($\mathcal{BL} \cup \mathcal{WL}$): $S'$ has
the same root, up to $F_{max}$ leaves and total cost of the leaves $\geq T_0-C$. There are three possibilities:
\begin{enumerate}
\item $\tilde p$ includes neither bad nor good addresses.
In this case, we can simply remove $\tilde p$ from $S$, \ie unblock $\tilde p$.
\item Only one of $\tilde p$'s children includes bad or good addresses.
In this case, we can replace $\tilde p$ with the child that contains the bad addresses.
\item Both of $\tilde p$'s children include bad or good addresses. In this case, $\tilde p$ is already
a longest common prefix and we do not need to do anything.
\end{enumerate}
Clearly, each of these operations transforms feasible solution $S$ into another
feasible solution with smaller or equal collateral damage while still preserving the capacity constraint. This is because the transformations filter the same amount of traffic, just using the longest prefix possible to do so. We can repeat this process for all prefixes that are in $S$ but not in LCP-tree($\mathcal{BL} \cup \mathcal{WL}$),
until we create a feasible solution $S'$ that includes only prefixes from LCP-tree($\mathcal{BL} \cup \mathcal{WL}$)
and has smaller or equal collateral damage.
\end{proof}

% FLOODING is a 2-dimensional knapsack problem (2KP) with an additional constraint, expressed by Eq. (\ref{P3-coverage}). %---that makes it harder.
%\footnote{2KP is a ``very hard'' problem: not only is it NP-Hard, but also the existence of a full polynomial-time approximation scheme for this problem would imply that $\mathcal{P}=\mathcal{NP}$ \cite{GensLevner}.}.

\begin{theorem} \label{theorem}
FLOODING (\ie Eq.(\ref{P3-OF})-(\ref{P3-coverage})) is NP-Hard.
\end{theorem}

\begin{proof}
% It is obvious that FLOODING is in $\mathcal{NP}$.
To prove that FLOODING is $\mathcal{NP}$-hard, we consider the knapsack problem with a cardinality constraint:
\begin{align}\label{1.5KP-OF} %\hspace{2.5cm}
\max \sum_{i \in N} p_{i} x_{i}&&\\
\label{1.5KP-cardinality}    \sum_{i \in N}  x_{i}  = k &&\\
\label{1.5KP-capacity} \text{s.t.} \sum_{i \in N} w_{i} x_{i}  \leq C_1
%\label{1.5KP-domain}      x_{i}  \in \{0,1\} 	\quad && \forall i \in \{1,..., N\}
\end{align}
which is known to be $\mathcal{NP}$-hard \cite{KPbook} and we show that it reduces to FLOODING.
To do this, we put FLOODING in a slightly different form, by making two changes.

First, we change the inequality in Eq. (\ref{P3-cardinality}) to an equality.
Any feasible solution to FLOODING that uses $F< F_{max}$ filters can be transformed to another feasible solution with
exactly $F_{max}$ filters, without increasing collateral damage. In fact, given a feasible solution $S$ that uses $F<F_{max}$ filters,
as long as $F < |\mathcal{BL}|$, it is always possible to remove a filter from a prefix $p$ and add two filters to the two prefixes
corresponding to $p$'s children in LCP-tree($\mathcal{BL} \cup \mathcal{WL}$).
The solution constructed this way uses $F+1$ filters, blocks all addresses blocked in $S$, and has a cost less or equal to $S$'s.
%In other words, we do not care how many filters a solution uses, as long as they are fewer than $F_{max}$.

Second, we define variables $\bar x_{p/l}= - x_{p/l} $, $\bar F_{max} = - F_{max}$ and  $\bar C = T_0- C$
and use them to rewrite FLOODING:
%\begin{align}\label{P3b-OF}
%\hspace{2.5cm}\max \sum_{l \leq r} g_{p/l} \bar x_{p/l}&&
%\end{align}
%\vspace{-0.3cm}
%s.t.
%\vspace{-0.3cm}
%\begin{align}
%\label{P3b-cardinality} \sum_{p/l} \bar x_{p/l} & = \bar F_{max}&&\\
%\label{P3b-capacity}    \sum_{p/l}\Big( g_{p/l} + b_{p/l}\Big) \bar x_{p/l} &\leq C && \\
%\label{P3b-coverage}   \sum_{p/l : ip \in p/l} x_{p/l} & \leq 1\quad&&\forall ip \in \mathcal{BL} &&
%\label{P3b-domain}      x_{p/l} & \in \{0,1\} 	\quad && \forall p/l
%\end{align}
\begin{align}\label{P3b-OF}
&\max \sum_{p/l} g_{p/l} \bar x_{p/l} \\
\text{s.t. } &\sum_{p/l} \bar x_{p/l} = \bar F_{max}, \\
\label{P3b-constraints} &\sum_{p/l}\Big( g_{p/l} + b_{p/l}\Big) \bar x_{p/l} \leq \bar C \\
\label{P3b-overlap} & \sum_{p/l : ip \in p/l}- \bar x_{p/l} \leq -1~\forall ip \in \mathcal{BL}&
\end{align}
%\vspace{-0.1cm}

For a given instance of the problem defined by Eq. (\ref{1.5KP-OF})-(\ref{1.5KP-cardinality}), we construct an equivalent instance of the problem
defined by Eq. (\ref{P3b-OF})-(\ref{P3b-overlap}) by introducing the following mapping.
For $ip \in \mathcal{BL}\cup \mathcal{WL}$: $\ g_{ip} = p_{i}$, $ (g_{ip} + b _{ip})=w_{i}$.
For all other prefixes $p/l$ that are not addresses in the blacklist or whitelist:  $(g_{p/l} + b _{p/l}) = \bar C + 1$. Moreover, we assign $\bar F_{max} = k$ and $\bar C = C_1$.
With this assignment, a solution to the problem defined by Eq. (\ref{1.5KP-OF}) can be obtained by solving FLOODING, then taking the values of variables
$x_{p/l}$ that are blocked.
\end{proof}

%{\b I. Greedy algorithm. \textcolor{blue}[include or omit this?]}
%Theoretical bound: $1/2$-approximation algorithm (XXX); performance in practical scenario are expected to be much better (\ie closer to optimal value).Running time: $O(N\log N)$.
%The basic idea is to associate with every prefix, $p/l$, the efficiency ratio $b_{p/l}/g_{p/l}$, which represent the amount of bad traffic that we block per every unit of good traffic accidentally blocked. The algorithm  %proceeds by \textcolor{blue}{iteratively select prefixes in decreasing efficiency.}
%\begin{algorithm}[h!]
% \caption{\label{Alg:Greedy-capacity}\label{Alg:P5} Greedy-capacity}
%\begin{algorithmic}[1]
%\begin{footnotesize}
%\STATE $\mathcal L  = \{ b_{p/l}/g_{p/l} \ \forall p/l\in \mathcal D\}$\\
%\STATE $TrafficOut = 0$\\
%\REPEAT\label{}
%\STATE find $(\bar l, \bar r) = \text{argmin}\mathcal L$. Add the corresponding filter
%\STATE remove all filters included in the larger filter $ (\bar l, \bar r) $
%\UNTIL{$TrafficOut >  \sum_{p/l}(g_{p/l} + b_{p/l}) - C$}
%\end{footnotesize}
%\end{algorithmic}
%\end{algorithm}

\subsubsection*{Algorithm}
Given the hardness of the problem, we do not look for a polynomial-time algorithm.
We design a pseudo-polynomial-time algorithm that optimally solves FLOODING,
Its complexity is linearly with the number of good and bad addresses and with the magnitude of $C_{max}$.

Our algorithm is similar to the one  that solves BLOCK-SOME, \ie, it relies on an LCP tree and a DP approach.
However, we now use the LCP tree of all the bad and good addresses.
Moreover, when we compute the optimal filter allocation for each subtree, we now need to consider not only the number of filters
allocated to that subtree, but also the corresponding amount of capacity (\ie, the amount of the victim's capacity consumed by
the unfiltered traffic coming from the corresponding prefix).
We can recursively compute the optimal solution bottom-up as before:
\begin{equation}
\label{P5-DP} z_p(F,c)=
%\begin{cases}
\min_{ \substack{n =0,...,F \\ m=0,...,c} }\{ z_{s_l}(F-n, c-m)+z_{s_r}(n, m) \} % & \text{if $t\leq c $,} \\
%g_p &\text{otherwise}
%\end{cases}
\end{equation}
%\begin{align}
%\label{P5-DP}
%z_p(F,c) = \min_{\substack{n = 0,...,F \\ m = 0,...,c}} \{z_{s_1}(F-n,c-m) +  z_{s_1}(n,m)\}
%z_p(F,c) = \min_{\substack{n =0,...,F \\ m=0,...,c-t}}\{  z_{s_l}(F-n, c-t-m) + z_{s_r}(n, m) \}
%\end{align}
where $z_p(F,c)$ is the minimum collateral of prefix p when allocating $F$ filters and capacity $c$ to that prefix.

\subsubsection*{Complexity}
Our DP approach computes $O(CF_{max})$ entries for every node in LCP-tree($\mathcal{BL} \cup \mathcal{WL}$).
Moreover, the computation of a single entry, given the entries of descendant nodes, require  $O(CF_{max})$ operations, Eq.(\ref{P5-DP}).
We can leverage again the observation that we do not need to compute $CF_{max}$ entries for all nodes in the LCP tree:
At a node $p$, it is sufficient to compute Eq.(\ref{P5-DP}) only for  $c= 0,...,\tilde C = \min\{ C, \sum_{ip \in p/l} w_{ip} \} \leq C$ and $f = 0,..., F=\max\{F_{max}, |leaves(p)|\}$.
Therefore, the optimal solution to FLOODING, $z_{root}(F_{max},C)$, can be computed in $O((|\mathcal{BL}|+|\mathcal{WL}|)C^2)$ time.
This is increasing linearly with the number of addresses in $\mathcal{BL} \cup \mathcal{WL}$ and is polynomial in $C$. The overall complexity is pseudo-polynomial because $C$ cannot be polynomially bounded in the input size. In the evaluation section, we present a heuristic algorithm that operates in increments $\Delta C$ of $C$. Finally, we note that $F_{max}<<C$ and thus $F_{max}$ does not appear in the asymptotic complexity.
  %The algorithm has pseudo-polynomial complexity since it is polynomial in $C$ that cannot be polynomially bounded in the input size.
 %\footnote{recall that the input size is measured as the number of bits necessary to encode a generic instance of the problem.}.

\subsubsection*{BLOCK-SOME vs. FLOODING}
There is an interesting connection between the two problems.
To see that, consider the partial Lagrangian relaxation of Eq. (\ref{P3-OF})--(\ref{P3-coverage}):
 \begin{align}\label{P3-OF-lagrangian}
\max_{\lambda \ge 0 } & \Big\{\min \sum_{p/l} \Big[ (1-\lambda) g_{p/l}  - \lambda b_{p/l}\Big]  x_{p/l} +  &&\\
&\nonumber ~~~~~~~~~~+ \sum_{p/l} \lambda T_0 - \lambda C\Big\}&&\\
\text{s.t.} &\sum_{p/l} x_{p/l} \leq F_{max} &&\\
 \label{P3-dual-coverage} & \sum_{p/l :  ip \in p/l} x_{p/l} \leq 1\quad \forall  \ ip \in \mathcal{BL} &&
%\label{P3-dual-domain}      \lambda \geq 0 &&
\end{align}
For every fixed $\lambda \geq 0$, Eq. (\ref{P3-OF-lagrangian})--(\ref{P3-dual-coverage}) are equivalent to Eq. (\ref{P2-OF})--(\ref{P2-coverage}) for a specific assignments of weights $w_{ip}$.
This shows that dual feasible solutions of FLOODING are instances of BLOCK-SOME for a particular assignment of weights.
The dual problem, in the variable $\lambda$, aims exactly at tuning the Lagrangian multiplier to find the best assignment of weights.

%\footnote{Problem (\ref{P3-OF-lagrangian})-(\ref{P3-dual-coverage}) can be solved in a standard way with a projected subgradient method \cite{KPbook}
%\begin{align}
%\label{P3-subg1} x^{(k)}_{p/l} &= x_{p/l}^*(\lambda^{(k)}), \ \forall p, l\\
%\label{P3-subg2} \lambda^{(k+1)} &= \big[ \lambda^{(k)} + \alpha_k \big(    \sum_{p/l}\Big( g_{p/l} + b_{p/l}\Big) - \sum_{p/l}\Big( g_{p/l} + b_{p/l}\Big)  x_{p/l})  - C       \big) \big]^+
%\end{align}
%where, $x_{p/l}^{(k)}$ is the $k$th iteratation, $x_{p/l}^*(\lambda^{(k)})$ is the optimal solution of (\ref{P3-OF-lagrangian})-(\ref{P3-dual-coverage}) for $\lambda = \lambda^{(k)}$,  $\alpha_k>0$ is the $k$th step size, and $[\cdot]^+$ indicates the projection over the set of non-negative numbers.}

%The interpretation of the dual problem, is that, among feasible solutions, it try to assign weights, such that, only the subset of the BL which causes the smallest CD and permit to fit all traffic within C is selected.

\subsection{DISTRIBUTED-FLOODING}

{\em Problem Statement:} Consider a victim $V$ that connects to the Internet through its ISP and is flooded by a set of attackers listed in a blacklist $\mathcal{BL}$, as in Fig.1(a).
To reach the victim, attack traffic  passed through one or more ISP routers. Let $\mathcal R$ be the set of unique such routers. % from some attacker to the victim.
Let each router $u\in \mathcal R$  have capacity $C^{(u)}$ on the downstream link (towards $V$) and a limited number of filters $F_{max}^{(u)}$. The volume of good/bad traffic through every router is assumed known. Our goal is to allocate filters on some or all routers, in a distributed way, so as to minimize the total collateral damage and avoid congestion on all links of the ISP network.

{\em Formulation.} Let the variables $x^{(u)}_{p/l}\in\{0,1\} $ indicate whether or not filter $p/l$ is used at router $u$. Then the distributed filtering problem can be stated as:
%\footnote{The constraint Eq.(\ref{P6-coverage}), can be changed summing not over all routers but only over %routers that are actually traversed by the malicious flow on its way to the victim. However, this does not %change the formulation and does not substantially decrease the complexity of the problem.}
%\vspace{-0.2cm}
\begin{align}
\label{P6-OF}
%\hspace{2.5cm}
\min \sum_{u\in \mathcal R}\sum_{p/l} g^{(u)}_{p/l} x^{(u)}_{p/l}&~ &&
\end{align}
%\vspace{-0.3cm}
%s.t.
%\vspace{-0.3cm}
\begin{align}
\label{P6-cardinality} \text{s.t.~}	\sum_{p/l} x^{(u)}_{p/l} & \leq  F^{(u)}_{max} \quad && \forall u\in \mathcal R
\end{align}
\begin{align}
\label{P6-capacity}    		T_0^{(u)} -\sum_{p/l}\Big( g^{(u)}_{p/l} + b^{(u)}_{p/l}\Big) x^{(u)}_{p/l} &\leq C^{(u)}  && \forall u\in \mathcal R\\
\label{P6-coverage}   	\sum_{u\in\mathcal R}\sum_{p/l \ni ip} x^{(u)}_{p/l} & \leq 1\quad&&\forall ip \in \mathcal{BL} %\\
%\label{P6-domain}      	x^{(u)}_{p/l} & \in \{0,1\} 	\quad && \forall p/l, u\in \mathcal R
\end{align}

{\em Characterizing an Optimal Solution}.
Given the sets $\mathcal{BL}$, $\mathcal{WL}$, $\mathcal R$, and $F_{max}^{(u)}$, $C^{(u)}$ at each router,  we have:
\begin{proposition}\label{prop:dist}
There exists an optimal solution of DIST-FLOODING
that can be represented as a set of $|\mathcal R|$ different pruned subtrees of the LCP-tree($\mathcal{BL}\cup \mathcal{WL}$),
%with the same root, at most $F_{max}$ leaves,  sum of costs associated to all subtree leaves is greater than $t0-C$,
each corresponding to a feasible solution of FLOODING for the same input, and s.t. every subtree leaf is not a node of another subtree.
\end{proposition}

{\em Proof.} Feasible solutions of DIST-FLOODING allocate filters on different routers s.t. Eq.(\ref{P6-cardinality}) and
(\ref{P6-capacity}) are satisfied independently at every router. In the LCP tree, this means having $|\mathcal R|$ subtrees, one for every router,
each having at most $F_{max}^{(u)}$ leaves and associated blocked traffic $\geq T^{(u)}_0 - C^{(u)}$, where $T^{(u)}_0$ is the total incoming traffic at router $u$. Each subtree can be thought as a  feasible solution of a FLOODING problem. Eq.(\ref{P6-coverage}) ensures that the same address is not filtered multiple times at different routers, to avoid waste of filters.
%We recall that the sum graph of two graphs $G_1$, $G_2$ is the graph with adjacency matrix given by the sum of adjacency matrices of $G_1$, and $G_2$.
%In other words,there is an edge in the sum graph of $G_1$, $G_2$ if and only that edge is present either in $G_1$ or $G_2$\footnote{as a consequence, if $G_1$ is a subgraph
%of $G_2$, the sum graph of these two graphs, is exactly $G_2$.}. Considering the sum graph of all subtree, a leaf in a subtree remains a leaf in the sum graph
%if and only if
In the LCP-tree, this translates into every leaf appearing at most in one subtree. \hfill $\blacksquare$
 %force not to use overlapping or identical filters at different routers.

{\em Algorithm.}
%Dual decomposition:
Constraint (\ref{P6-coverage}), which imposes that different routers do not block the same prefixes,  prevents us from a direct decomposition of the problem.   To decouple the problem, consider
the following partial Lagrangian relaxation:
%\begin{align}
%\nonumber L(x, \lambda) =~~~~~~~~~~~~~~~~~~~~~~~~~~~~~~&&\\
%\vspace{-0.1cm}
\begin{align}
\nonumber L(x, \lambda) = &\sum_{u\in \mathcal R}\sum_{p/l} g^{(u)}_{p/l} x^{(u)}_{p/l} + \sum_{ip \in\mathcal{BL}}\lambda_{ip} \Big( \sum_{u\in\mathcal R}\sum_{p/l \ni ip} x^{(u)}_{p/l}  - 1\Big) \\
% $$= \sum_{u\in \mathcal R}\Big(  \sum_{p/l\in T_u} g^{(u)}_{p/l} x^{(u)}_{p/l}  +  %\sum_{ip \in\mathcal{BL}}\lambda_A\sum_{p/l \ni ip} x^{(u)}_{p/l} \Big) - \sum_{ip \in\mathcal{BL}} \lambda_A $$
 =& \sum_{u\in \mathcal R}\Big(  \sum_{p/l} \Big(g^{(u)}_{p/l}  + \lambda_{p/l} \Big) x^{(u)}_{p/l}\Big)    - \sum_{ip \in\mathcal{BL}} \lambda_{ip}
\end{align}
where $\lambda_{ip}$ is the Lagrangian multiplier (price) for the constraint in Eq.(\ref{P6-coverage}), and
$\lambda_{p/l} =  \sum_{ip \in p/l}\lambda_{ip}$ is the price associated with prefix $p/l$. With this relaxation, both the objective function and the other constraints immediately decompose in  $|\mathcal R|$ independent sub-problems, one per router $u$:
\begin{align}\label{P6sub-OF}
%\hspace{2.5cm}\
\min \sum_{p/l} \Big(g^{(u)}_{p/l} + \lambda_{p/l}\Big)x^{(u)}_{p/l}&&
\end{align}
\begin{align}
\label{P6sub-cardinality} 	\text{s.t.} \sum_{p/l} x^{(u)}_{p/l} &  \leq F_{max}^{(u)} \quad && \\
\label{P6sub-capacity}    		T_0^{(u)} - \sum_{p/l}\Big( g^{(u)}_{p/l} + b^{(u)}_{p/l}\Big)  x^{(u)}_{p/l} &\leq C^{(u)}  && %\\
%\label{P6sub-domain}      	x^{(u)}_{p/l} & \in \{0,1\} \quad && \forall p/l
\end{align}
The dual problem is:
\begin{align}
\label{P6master-OF}
%\hspace{2.5cm}
\max_{\lambda_{ip} \geq 0} \sum_{u\in \mathcal R}h_u(\lambda)  - \sum_{ip \in\mathcal{BL}} \lambda_{ip}
\end{align}
%\vspace{-0.3cm}
%s.t.
%\vspace{-0.3cm}
%\begin{align}
%\label{P6master-domain}      	\lambda_{p/l} & \geq 0 	&&
%\end{align}
where $h_u(\lambda)$ is the optimal solution of (\ref{P6sub-OF})-(\ref{P6sub-capacity}) for a given $\lambda$. Given the prices $\lambda_{ip}$, every sub-problem (\ref{P6sub-OF})-(\ref{P6sub-capacity}) can be solved independently and optimally by router $u$ using \eg Eq. (\ref{P5-DP}). Problem (\ref{P6master-OF}) can be solved using a projected sub-gradient method %similarly to Eq.(\ref{P3-subg1})-(\ref{P3-subg2}),
 \cite{KPbook}.
In particular, we use the following update rule to compute shadow prices at each iteration:
$$\lambda^{(k+1)}_{ip} = \lambda^{(k)}_{ip} + \alpha( \sum_u \sum_{ p/l \ni \ ip} x_{p/l}^{(u)} - 1)$$
where $\alpha$ is the step size. The interpretation of the update rule is quite intuitive: for every $ip$ that is filtered with multiple filters the corresponding shadow price, $\lambda_{ip}$, is augmented proportionally to the number of times it is blocked.
Increasing the prices has in turn the effect of forcing the router to try to unblock the corresponding $ip$. The price is increased until a single filter is used to block that $ip$.

Note, however, that since $x$ is an integer variable, $x \in\{0,1\}$, the dual problem is not always guaranteed to converge to a primal feasible solution \cite{boyd}.

{\em Distributed vs. Centralized Solution.}
The above formulation lends itself naturally to a distributed implementation.
Each router needs to solve only their own sub-problem
(\ref{P6sub-OF})-(\ref{P6sub-capacity}) independently from others. A single machine (\eg the victim's gateway
or a dedicated node) should solve the master problem (\ref{P6master-OF}) %-(\ref{P6master-domain}),
to iteratively find the prices that coordinate all sub-problems. At every iteration of the sub-gradient, the new $\lambda_{ip}$'s need to be broadcasted to all routers. Given the $\lambda_{ip}$'s, the routes solve independently a sub-problem each and return the computed $x_{p/l}^{(u)}$ to the node in charge of the master problem.
%Equivalent formulations of DISTR-FLOODING, other than the one in Eq.(\ref{P6-OF})-(\ref{P6-domain}), are also possible and  may generate %different decomposition of the problem. \textcolor{blue}{The hope is to try to reduce the amount of communication required.}
Even in a centralized setting, our scheme lends itself to  parallel computation
of Eq.(\ref{P6-OF})-(\ref{P6-coverage}).

\section{\label{sec:simulations}Practical Evaluation}

\begin{figure}[t!]
\begin{center}
\includegraphics[scale=0.35 ]{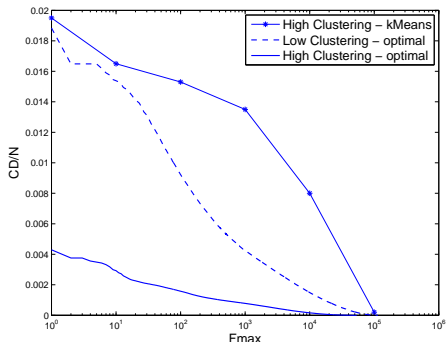}
\vspace{-15pt}
\caption{\label{fig:block-all} {\bf Evaluation of BLOCK-ALL in Scenarios I and II} in terms of collateral damage (CD) (normalized over the number of malicious sources $N$) vs. number of filters $F_{max}$. We compare Algorithm 1 to K-means. In particular, we simulated 50 runs of Lloyd's heuristic to solve the K-means problem in order to avoid local minima. We also run Algorithm 1 on two traces, those with the highest and lowest degree of clustering.
``High clustering'' and ``Low clustering'' refers to the two example blacklists in Scenarios I and II, respectively. }
\end{center}
\vspace{-20pt}
\end{figure}

In this section, we evaluate our algorithms using real logs from malicious traffic.
We demonstrate that our algorithms bring significant benefit compared to non-optimized filter selection or to generic clustering algorithms, in a wide range of scenarios. The reason behind this benefit is the well-known fact that sources of malicious traffic exhibit spatial and temporal clustering \cite{uncleanness, clustering, mao2006analyzing, ramachandran2006understanding, venkataraman2007exploiting, xie2007dynamic,zhang-highly}, which is exploited by our algorithms. Indeed, clustering in a blacklist allows to use a small number of filters to block prefixes with high density of malicious IPs at low collateral damage. Furthermore, it has also been observed that the good and bad addresses are typically not co-located, which allows for distinguishing between good and bad traffic
\cite{ramachandran2006understanding, venkataraman2007exploiting, estan}, and in our case
for efficient filtering of the most ``contaminated'' prefixes.

%Due to lack of space, the simulations presented in this section are not exhaustive.  However, they demonstrate the above point as well as some of the structural properties of the solution for BLOCK-ALL and BLOCK-SOME, which are at the heart of this framework. As discussed in section \ref{sec:algorithms}, FLOODING is essentially an instance of BLOCK-SOME for a particular assignment of weights and DIST-FLOODING consists of several FLOODING problems.

\subsection{Simulation Setup}

% 1) what kind of data do we have
% 2) in the data tehre is clustering
% 3) clustering is good
% 4) good traffic
%In this section we show some of the results obtained by running the above optimal algorithms on real distribution of malicious IP sources.
We used 61-day logs from {\tt Dshield.org} \cite{dshield} - a repository of firewall and intrusion detection logs collected.
The dataset consists of 758,698,491 attack reports, from 32,950,391 different IP sources belonging to about 600 contributing organizations. Each report includes a timestamp, the contributor ID, and the information for the flow that raised the alarm, including the (malicious) source IP and the (victim) destination IP.  Looking at the attack sources in the logs, we verified that malicious sources are clustered in a few prefixes, rather than uniformly distributed over the IP space, consistently with what was observed before \eg in \cite{clustering, uncleanness,mao2006analyzing, ramachandran2006understanding, venkataraman2007exploiting}.

In our simulations, we considered a blacklist to be the set of sources attacking a particular organization (victim) during a day-period. The degree of clustering varied significantly in the blacklists of different victims and across different days. The higher the clustering, the more the benefit we expect from our approach.
% We also generated good traffic according to a realistic scenario: a domain hosting 20 servers, each server with average rate of 1,000  incoming good connections per second, each connection generating 5KB of traffic.
We also simulated the whitelist, by generating good IP addresses according to the multifractal distribution in \cite{kohler} on routable prefixes. We performed the simulations  on a linux-machine with 2.4 GHz processor with 2 GB RAM.
%Although the algorithms presented can be easily parallelized, we decided we carry all computations using  at a single core processor in order to ...

\subsection{Simulation of BLOCK-ALL and BLOCK-SOME}

% FIGURE: block-all (on single victims):
% 1) describe the curves
% 2) much better than k-means
% 3) high clustering is better
% 4) however, when low clustering CD decrease rapidly with Fmax

{\em Simulation Scenarios I \& II.} In Fig. \ref{fig:block-all}, we considered two example blacklists corresponding to two different victims, each attacked by a large number (up to 100,000) of malicious IPs in a single day. We picked to present the blacklists with the highest and the lowest degree of source clustering observed in the entire dataset these two in particular, referred to as ``High Clustering'' (I) and ``Low Clustering'' (II), respectively; these demonstrate the range of benefit from our approach.
% and \textcolor{blue}{good traffic generated according to a realistic scenario: a domain hosting 20 servers, each server with average rate of 1,000  incoming good connections per second, each connection generating 5KB of traffic.}

\begin{figure}[t!]
\begin{center}
\includegraphics[scale=0.40 ]{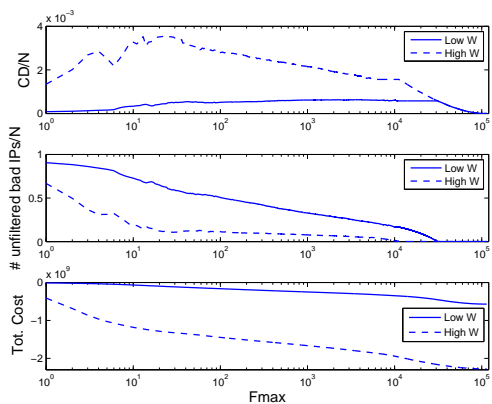}
\vspace{-15pt}
\caption{\label{fig:block-some} {\bf Evaluation of BLOCK-SOME for Scenario II (Low Clustering blacklist)}. Three metrics are considered: (a) Collateral damage (CD) (b) Number of unfiltered bad IPs (UBIP) (c) Total cost $CD+W\cdot UBIP$. The operator expresses preference for  UBIP vs. CD by tuning the weight $W=\frac{w_b}{w_g}$. We considered two values of $W$: a higher ($2^{14}$) and a lower ($2^{10}$) one.}
%\textcolor{blue}{high $W=2^{14}$, low $W=??$}
\end{center}
\vspace{-20pt}
\end{figure}

{\em BLOCK-ALL.} We ran Algorithm 1 in these two scenarios and we show the results in Fig.\ref{fig:block-all}. We made the following observations.
First, the optimal algorithm performs significantly better than a generic clustering algorithm that does not exploit the structure of IP prefixes.
In particular, it reduces the collateral damage (CD) by up to $85$\% compared to K-means, when run on the same (high-clustering) blacklist.
Second, the degree of clustering in a blacklist matters: the CD is lowest (highest) in the blacklist with highest (lowest) degree of clustering, respectively.
Results obtained for different victims and days were similar and lied in between the two extremes. A few thousands of filters were sufficient to significantly reduce collateral damage in all cases.

% FIGURE: block-some (on single victims):
%\begin{figure}[th!]
%\begin{center}
%\includegraphics[scale=0.45 ]{figs/block-some/block-some-subplot2.eps}
%\caption{\label{fig:block-some2} a}
%\end{center}
%\end{figure}

% 1) describe the curves
% 2) better than BLOCK-ALL. CD are smaller, and few unfiltered. CD(block-some) = CD(block-all) only when unfiltered = 0. This is expected
% 3) the tradeoff depend on W. High W: high CD, and low unfiltered.
% 4) Curves depends on clustering,  and W. However, some common feature:
% 4i) CD increase and then decrease (on avg)
% 4ii) unfiltered: decrease
% 5) explain 4) and the last part of the curves.

{\em BLOCK-SOME.} In Fig. \ref{fig:block-some}, we focus on Scenario II, \ie the Low Clustering blacklist and thus the highest CD (dashed line in Fig.\ref{fig:block-all}), which is the least favorable input for our algorithm. Unlike BLOCK-ALL, BLOCK-SOME allows the operator to trade-off lower CD for some unfiltered bad IPs by appropriately tuning the weights. % to single IP or prefix (either good or bad).
 For simplicity, in Fig. \ref{fig:block-some}, we assigned the same weights $w_g$ and $w_b$ to all good and bad sources;
however, the framework has the flexibility to assign different weights to different IPs. In Fig. \ref{fig:block-some}(a), the CD is always smaller than the corresponding CD in Fig. \ref{fig:block-all}; they become equal only when we block all bad IPs. In Fig. \ref{fig:block-some}(b), we observe that BLOCK-SOME reduces the CD by 60\% compared to BLOCK-ALL while leaving unfiltered only 10\% of bad IPs and using only a few hundreds a filters.
%Eg. in case of critical attack we may accept to suffer higher CD but we want to

%\begin{figure*}[t!]
%\centering
%\subfigure[Actual network.]
%{\includegraphics[width=1.5in]{./figs/overview.eps}}
%\hspace{50pt}
%\subfigure[Hierarchy of source IP addresses and prefixes]
%{\includegraphics[width=2.1in]{./figs/overviewIPspace.eps}}
%\caption{Example of a distributed attack. Let's assume that the gateway router $G$ has only two filters available to block malicious traffic and protect the %victim $V$. It uses $F1$ to block a single malicious address (A) and $F2$ to block prefix $a.b.c.*$, which contains 3 malicious sources but also one legitimate source (B). Therefore, the selection of filter $F2$ trades-off the collateral damage (blocking B) for the reduction in the number of filters (from 3 to 1).
%\label{fig:overview}}
%\vspace{-10pt}
%\end{figure*}

\begin{figure*}
\centering
\subfigure[CD/N vs Fmax]
{\includegraphics[width=2.35in]{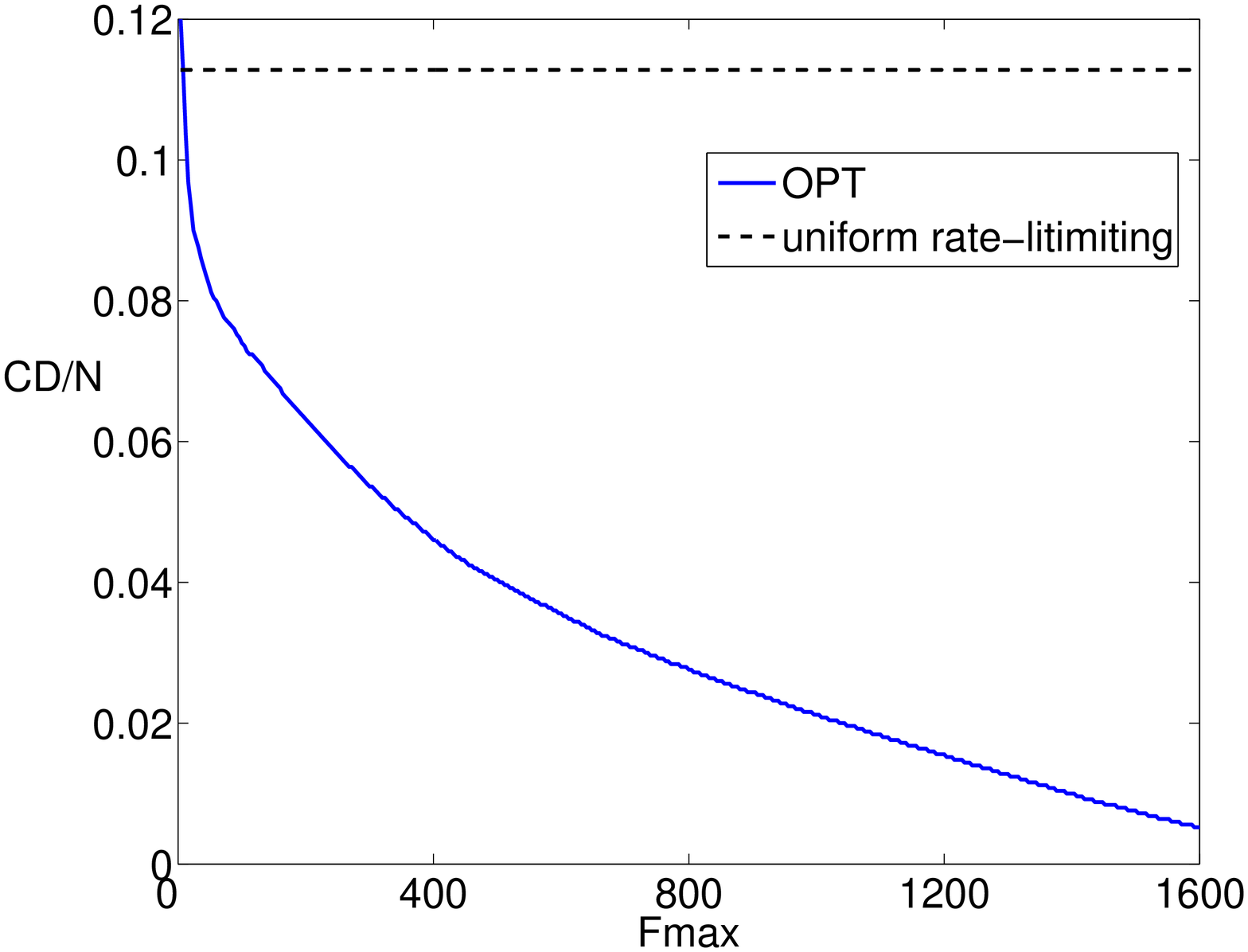}}
\subfigure[CD/N vs Cmax]
{\includegraphics[width=2.3in]{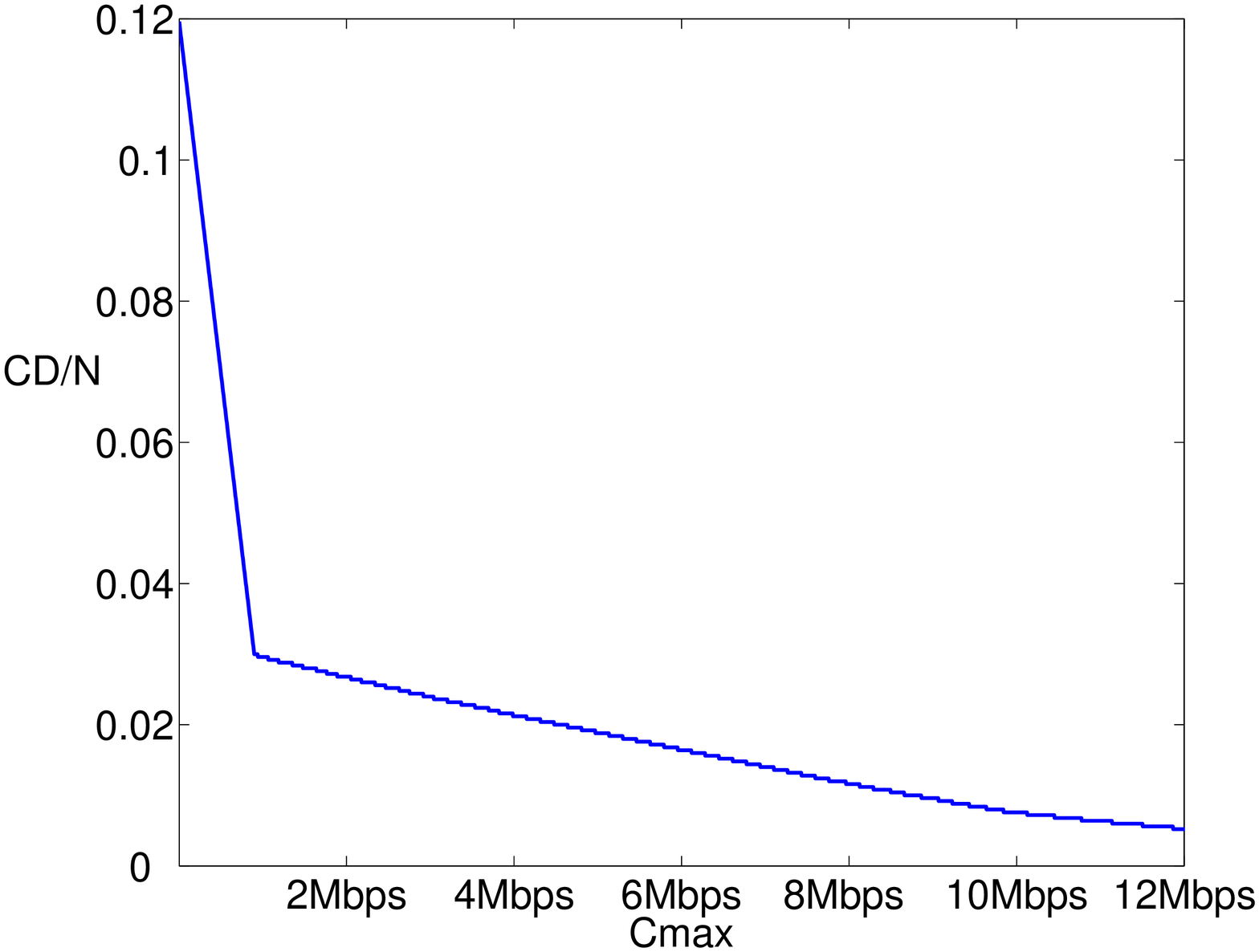}}
\subfigure[Fmax vs Cmax]
{\includegraphics[width=2.35in]{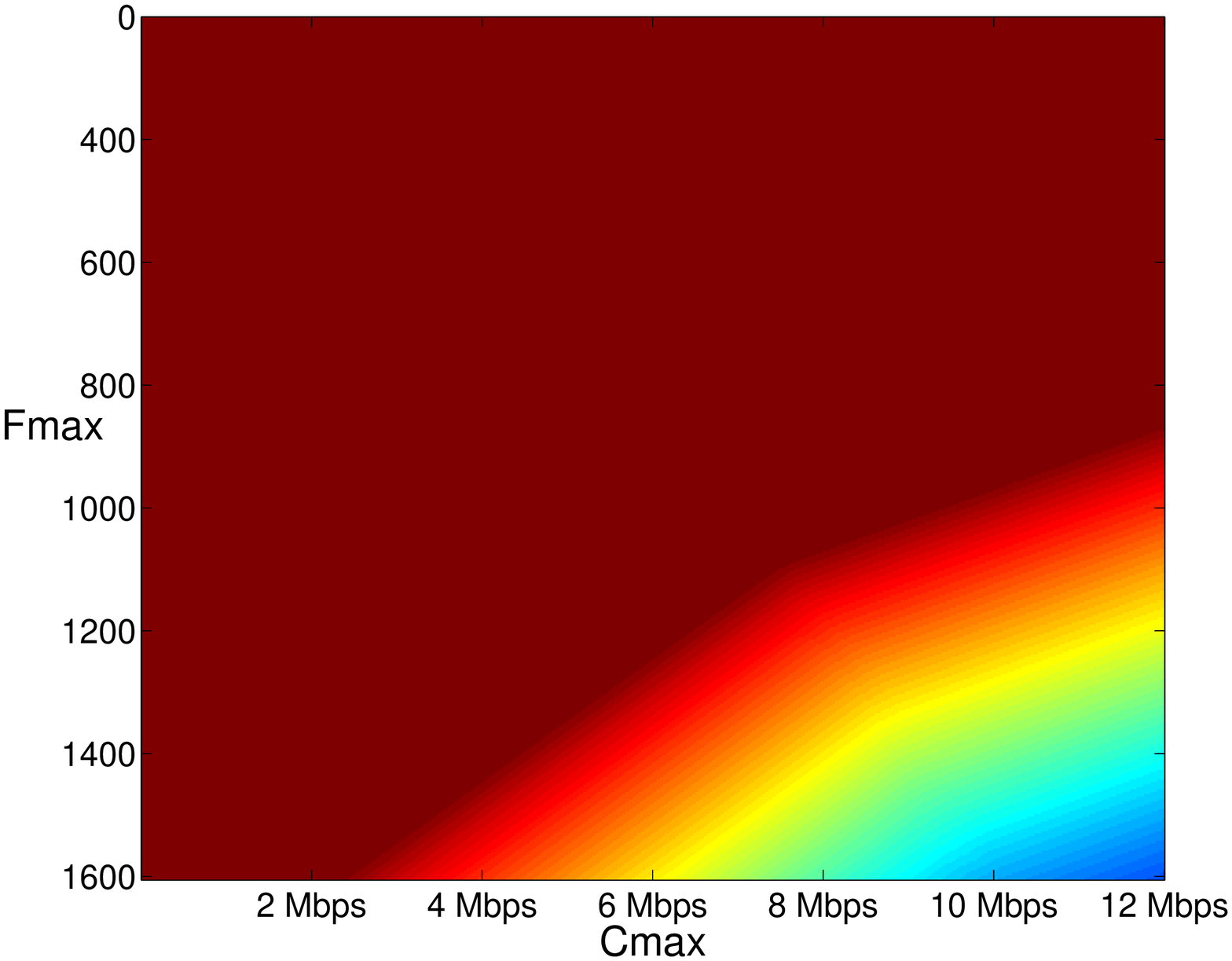}}
\caption{ {\bf Optimal Solution of FLOODING in Scenario III.} In (a), we show the normalized collateral damage, $CD/N$, as a function of the number of available filters, $F_{max}$, when $C$ is fixed, $C=C_{max}$. In (b), we show  $CD/N$ as a function of the available capacity, $C_{max}$, when $F$ is fixed, $F=F_{max}$. In (c), we show how the ratio $CD/N$ varies as a function of both $C_{max}$ and $F_{max}$. \label{fig:flooding}}
\vspace{-10pt}
\end{figure*}

In Fig. \ref{fig:block-some}(c), the total cost of the attack (\ie the weighted sum of bad and good traffic blocked) decreases as $F_{max}$ increases.
The interaction between these two competing factors is complex and strongly depends on the input blacklist and whitelist.
In the data we analyzed, we observed that CD tends to first increase and then decrease with $F_{max}$, while the number of unfiltered bad IPs tends to decrease
%\footnote{We can explain this as follows. As more filters become available, the new optimal solution can be constructed by (i) blocking a new cluster of bad IPs, (ii) splitting a blocked cluster into two and assigning two filters, or (iii) a combination of (i)\&(ii)\& merging of existing filters. For small $F_{max}$, option (i) is dominant: the inherent clustering of blacklists allows finding a cluster that is not blocked yet. In this case, the maximum reduction of the attack cost is obtained when a filter is allocated to the cluster of bad IPs. This reduces the number of unfiltered bad IPs, but might also increase the overall CD. When this is not possible, option (ii) becomes dominant: CD decreases and the number of unfiltered bad IPs remains constant or decreases slowly.}.
 The ratio $w_b/w_g $ captures the effort\footnote{Since we picked a ratio $w_b/w_g > 1$, bad IPs are more important. When $F_{max}$ is high, the algorithm first tries to cover small clusters or single  bad IPs. In the case of high $W$, this happens around $10,000$ filters. CD remains almost constant in this phase, at the end of which all bad IPs are filtered (as in Fig.\ref{fig:block-some}(b)). In the final phase, the algorithm releases single good IPs, which are less important and all bad IPs are blocked similarly to BLOCK-ALL.} made by BLOCK-SOME to block all bad IPs and become similar to BLOCK-ALL.

%\item figures:  CD and \#unfiltered IPs (two y-axis) vs Fmax  (on x-axis). 6 or 9 curves: CD vs Fmax, \#unfiltered IPs vs. Fmax, totCost vs. Fmas, for 2 or 3 different values of  W \footnote{only one victim ? }.
%\item conclusion:  assignements of weight give a lot of flexibility. We can tune the ratio CD/\#unfiltered IPs: if block-all although optimal still has too high CD we can tune the trade-off CD vs. unfiltered IPs as we prefer \footnote{note: block some degenerate in block-all for large weight to bad IPs and small weights to good. Maybe we shuold put this in the formulation. \textcolor{blue}{[yes, I will]}} by changing the assignment of weights. The flexibility provided by the assignment of weights goes
%also beyond that and permit to block or allow specific prefixes the network operator may be interested in.
%\item \textcolor{blue}{I would say only 2 values of W are enough. If it is too messy, you can put the 3 curved per W in a different subfigure of the same figure}

% FIGURE: block-flooding (on single victims):
%\item TODO. \textcolor{blue}{compare to random dropping}

\subsection{Simulation of FLOODING and DIST-FLOODING}

%Due to lack of space, the simulations presented in this section are not exhaustive.  However, they demonstrate the above point as well as some of the structural properties of the solution for BLOCK-ALL and BLOCK-SOME, which are at the heart of this framework. As discussed in section \ref{sec:algorithms}, FLOODING is essentially an instance of BLOCK-SOME for a particular assignment of weights and DIST-FLOODING consists of several FLOODING problems.

%As shown in the previous sections, BLOCK-ALL and BLOCK-SOME problems  are at the heart of this framework:
%FLOODING is essentially an instance of BLOCK-SOME with an additional capacity constraint which makes the problem hard, while
% DIST-FLOODING can be decomposed into several FLOODING problems.

%In this section, we evaluate  the performance of FLOODING and DIST-FLOODING.

{\em Simulation Scenario III.} We consider a web server under a DDOS attack. We assume that the server has a typical access bandwidth of $C=100$ Mbps and can handle 10,000 connections per second (a typical capability of a web server that handles light content). We assume that each good (bad) source generates the same amount of good (bad) traffic. We also assume that $F_{max}=12,000$ filters are available (consistently with the discussion in footnote 1) and we vary $F=1,...F_{max}$. 
%- C=100Mbps is the typical capacity that comes up when you google for "unmetered servers"
%- 10,000 connections/sec is the typical capability of a web server that handles light content
%- 12,000 filters is somewhat arbitrary, but I guess what matters is not the exact number, but how many filters vs attackers we have, so it's OK to start from somewhere.
Before the attack, $5,000$ good sources are picked from \cite{kohler} and utilize 10\% of the capacity. During the attack, the total bad traffic is  $10C=1$Gbps and is generated by
 a typical blacklist (141,763 bad source IPs), based on {\tt Dshield} logs of a randomly chosen victim for a randomly chosen day.\footnote{However, because this problem is NP-hard  we do not simulate the entire IP space, but the range $[60.0.0.0, 90.0.0.0]$,  which is known to account for the largest amount of malicious traffic, \eg see \cite{ramachandran2006understanding}. We also scale  all parameters by a factor of 8, $F_{max}, C_{max}, w_{i}$ to maintain a constant ratio between the number of IPs and $F_{max}$, and, the total flow generated and $C_{max}$.}

%In the following, we will see how it is possible to modify the optimal algorithm for FLOODING to achieve the desired tradeoff between performance and scalability.

{\em FLOODING - Optimal.}
Fig. \ref{fig:flooding}(a) and Fig. \ref{fig:flooding}(b) show the  collateral damage of the optimal solution of FLOODING, for Scenario III, as a function of the number of available filters $F_{max}$, and as a function of the bottleneck capacity $C$, respectively. As baseline for comparison, we simulate {\em uniform rate-limiting}, which drops the same fraction of all incoming connections and is a common practice in DDOS attacks. Since bad sources outnumber the good sources in a typical DDOS attack, uniform rate-limiting  penalizes disproportionally the good sources.
While this solution is always applicable and requires only one rate limiter, more filters  can drastically reduce the collateral damage.

We also observe that varying the number of filters or the available capacity has a different impact on the collateral damage.
While the collateral damage decreases exponentially as the number of filters increases, when we increase the available capacity we observe two trends.
First, as capacity increases, the optimal solution allows traffic from good sources that do not belong in prefixes with many malicious sources.
This causes a linear decrease with slope equal to the amount of traffic generated by good sources.
For even larger $C$, good IPs located in the same prefix as malicious sources are released.
This  trend depends on the specific clustering of good and bad IPs considered as well as on the amount of traffic generated by both good and bad sources.
% The turning point between these two trends is usually clearly distinguishable and depends again on the specific data considered.
In  Fig. \ref{fig:flooding}(c) we plot the collateral damage as a function of both the number of available filters and the available capacity.
When the value of $F$ ($C$) is too low, increasing $C$ ($F$) does not yield any benefit.
Most of the improvement is obtained when both resources increase.

\begin{figure}[t!]
\begin{center}
\includegraphics[width=3in]{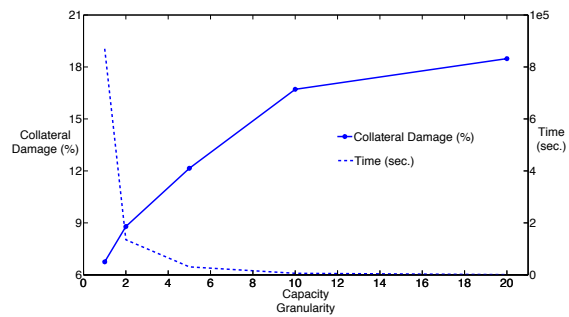}
\caption{\label{fig:flooding-heu} {\bf Heuristic solution of FLOODING in Scenario IV.} An approximate  solution  (higher CD)  is obtained by solving only the sub-problems $z_p(f,n\Delta C)$ for   $n\in \mathbb N$ and $C=n\Delta C \leq C_{max},~n=1,2...$. The coarser the capacity increments $\Delta C$, the fewer sub-problems we need to solve, but at the cost of higher collateral damage. In this scenario, increasing $\Delta C$ significantly reduces the computational time by 3 orders of magnitude, while the percentage of good traffic that is blocked (CD \%) is only increased by a factor 3.}
\end{center}
\vspace{-15pt}
\end{figure}

{\em FLOODING - Heuristic.} %As shown in Fig. \ref{fig:flooding}, the optimal solution of FLOODING significantly reduces the total cost of the attack on real traces.
The benefit of the optimal solution of FLOODING comes at high computational cost, due to the intrinsic hardness of the FLOODING problem. %, Sec. \ref{sec:flooding}.
To address this issue, we design a heuristic for solving FLOODING, which can be tuned to achieve the desired tradeoff between collateral damage and computational time.
In particular, instead of solving all subproblems, $z_p(f,c)$, for all possible values of $f\leq F_{max}$ and $c\leq C_{max}$, we consider discrete increments of capacity $c=n\Delta C$, with step size $\Delta C$.
% only subproblems such that  $ c \leq C_{max}$ and $c =  n \Delta C$  for some $n \in \mathbb N$ and $\Delta C>0$.
%In other words, we consider only discrete increments of $c$ of size $\Delta C$.
If $\Delta C = \min\{ w_{ip} \}$, the finest granularity of $c$ is considered, and the problem is optimally solved.
If $\Delta C > \min\{ w_{ip} \}$ we may get a sub-optimal solution, but we reduce the computation cost, as fewer iterations are required to solve the DP.

{\em Simulation Scenario IV.} We consider again a DDOS attack launched by 61,229 different bad source IPs, based on the {\tt Dshield.org} logs. The available capacity, $C=100$ Mbps. Before the attack, the legitimate traffic consumes $\frac{1}{2}C=50$Mbps. During the attack, the total bad traffic generated is $100C=10$Gbps. This scenario is more challenging than scenario III, because there is less unused capacity before the attack, and more malicious  traffic during the flooding attack.

In Fig. \ref{fig:flooding-heu},  we show the percentage of good traffic that is blocked by the heuristic vs. the time required to obtain a solution, for scenario IV.
As we can see in Fig. \ref{fig:flooding-heu}, the optimal solution of FLOODING ($\Delta C=1$) requires about 1 day of computation and has CD that is only 6\% of the total good traffic.
Larger values of $\Delta C$ allows to dramatically reduce the computational time by about 3 orders of magnitude while the CD  is only increased by a factor 2-3. This asymmetry was also the case in other {\tt Dshield} logs we simulated. This can be very useful in practice: an operator may decide to use an approximation of the optimal filtering policy to immediately cope with incoming bad traffic,
%opt for an approximation of the optimal solution, with slightly higher CD, to immediately cope with the incoming bad traffic,
and then successively refine the allocation of filters to further reduce the collateral damage if the attack persists. % over time.

\begin{figure}
\begin{center}
\includegraphics[height=1.4in ]{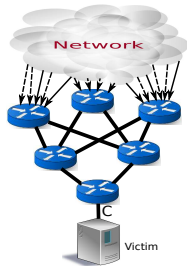}
\caption{\label{fig:topology} {\bf Distributed Flooding.}
 This topology exemplifies the part of a potentially larger ISP topology involved in routing and blocking traffic towards victim $V$.
 %that needs to be traversed, by malicious and legitimate traffic, to reach a  victim.
 The edge routers receive all incoming, malicious and legitimate, traffic towards victim $V$ and route it through shortest-paths  with ties broken randomly. Any of the traversed routers (indicated with circles) can be used to deploy ACLs and block the malicious traffic.}
\end{center}
\vspace{-15pt}
\end{figure}

{\em DISTRIBUTED-FLOODING.} We simulated the scenario where an ISP utilizes multiple routers to collaboratively block malicious traffic. We consider the same scenario (III) as for the optimal flooding for a single router, but now we assume that the traffic reaches the victim routed over the example topology illustrated in Fig.\ref{fig:topology}.

We use a sub-gradient descent method to solve the dual problem in Eq.(\ref{P6master-OF}). In Fig.\ref{fig:dist-flooding}, we show the convergence of the method for two different step sizes: $0.05$ and $0.01$. We also compare against the ``no coordination'' case, when routers do not coordinate  but act independently to block malicious traffic; this corresponds to the first iteration of the sub-gradient method. In the next iterations, routers coordinate, through the exchange of shadow prices $\lambda$, and avoid the redundant overlap of prefixes at multiple locations. This reduces the collateral damage significantly, \ie by $\sim 50$\%.
% shadow prices corresponding to the overlap of filters deployed at different routers

%\textcolor{blue}{The communication overhead of this scheme is limited. At each iteration, the routers need to exchange the value of the shadow prices, $\lambda_{ip}$, with a central controller or with eachother.For 100,000 of bad IPs, if we encode the value of variables $\lambda_{ip}$ using 2 Bytes, each router need to send about 200KB of data at each iteration. For 80 iterations of the subgradient algorithm, this leads to a total 16MB of information exchange.}

\begin{figure}
\begin{center}
\includegraphics[width=2.6in]{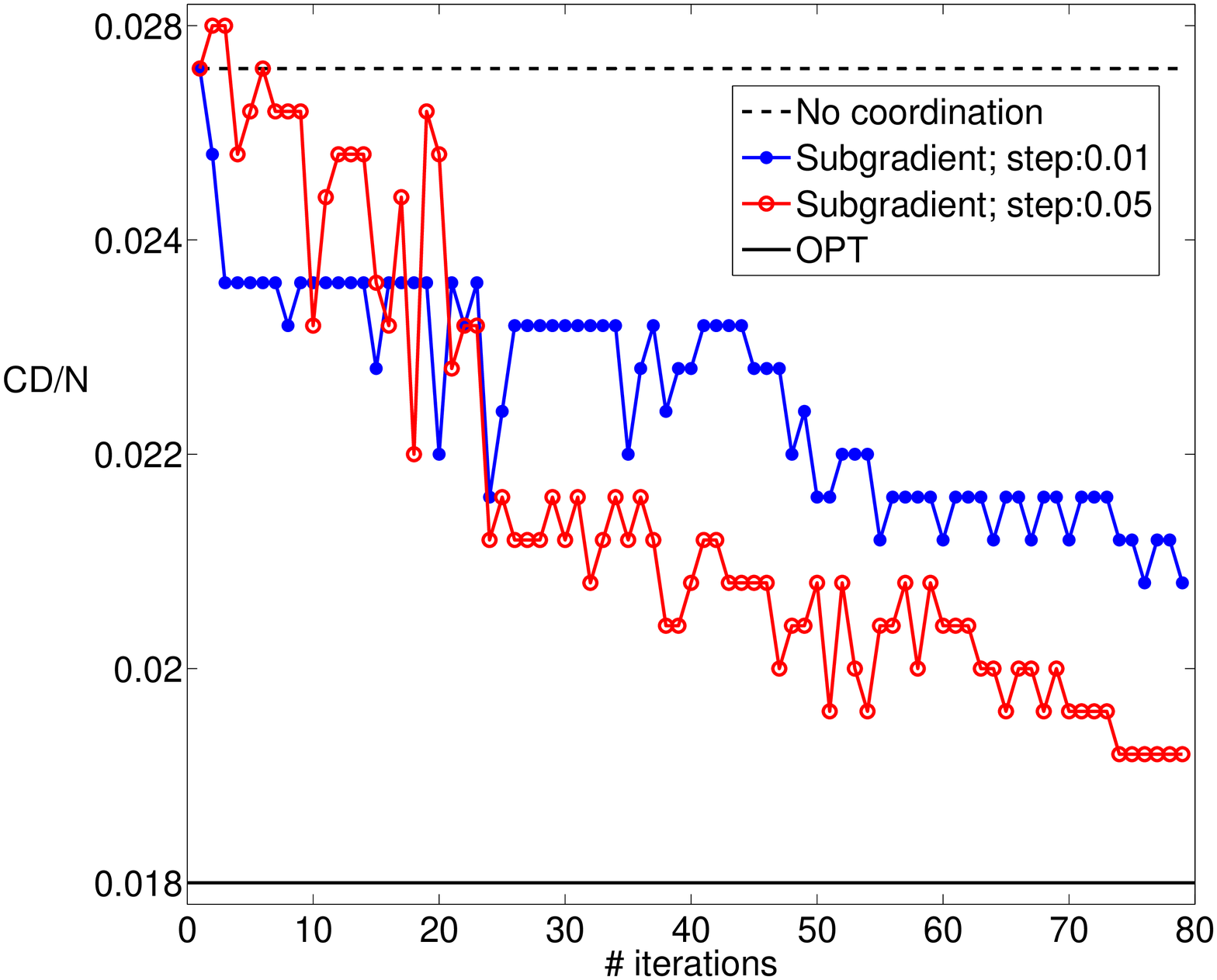}
\caption{\label{fig:dist-flooding} {\bf Evaluation of DIST-FLOODING in Scenario III.}  Results are shown for the distributed algorithm and two different values for the step size (0.01 and 0.05) of the subgradient method. The dashed line shows the case of ``No Coordination'', \ie when each router acts independently. }
\end{center}
\vspace{-10pt}
\end{figure}

\section{\label{sec:related}Our Work in Perspective}

%\subsection{\label{sec:assumptions}The Bigger Picture of Defense against Malicious Traffic}

{\bf Bigger Picture and Assumptions.}
Dealing with malicious traffic is a hard problem that
requires the cooperation of several components, including detection and mitigation techniques,
as well as architectural aspects.
In this paper, we do not propose a novel solution. Instead,
we optimize the use of filtering - a mechanism that already exists on the Internet
today and is a necessary building block of any bigger solution. We focus on the optimal construction
of filtering rules, which can be then installed and propagated by filtering protocols \cite{aitf, xiaowei2}.

We rely on  an intrusion detection system
or on historical data, to distinguish good from bad traffic and to provide us with a blacklist.
Detection of malicious traffic is an important problem but out of the scope of this paper.  The sources of legitimate traffic are also assumed known and used for assessing the
collateral damage; \eg web servers or ISPs typically keep historical data and know their important customers.

We also consider addresses in the blacklist to not be spoofed. This is reasonable today that attackers use botnets, and control a huge number of infected hosts for a short period of time, so that they do not even need to use spoofing. On 2005, less than 20\% of addresses were spoofable \cite{spoofer}, while in 2008, only 7\% of addresses in Dshield logs were found likely spoofed \cite{clustering}. Even if there is some amount of spoofed traffic, our algorithms treat it as the rest of malicious traffic and weight the cost vs. the benefit of blocking a source prefix (which may include both malicious spoofed and legitimate traffic). Looking into the future, there is also a number of proposals promising to enforce source accountability, including ingress filtering \cite{ingress}, self-certifying addresses \cite{AIP}, packet passports\cite {passport}. To the extent that spoofing interferes with the ability to define blacklists, our algorithms work best together with an anti-spoofing mechanism, but also do the best that can be done today without it.

{\bf Deployment scenario.} A practical deployment scenario is that of a single network under the same administrative authority, such as an ISP or a campus network. The operator can use our algorithms to install filters at a single edge router or at several routers, in order to optimize the use of its resources and to defend against an attack in a cost-efficient way. Our distributed algorithm may also be useful, not only for a routers within the same ISP, but also, in the future, when different ISPs start cooperating against common enemies.

{\bf ACLs vs. firewall rules.} Our algorithms may also be applicable in a different context: to configure firewall rules to protect public-access
networks, such as university campus networks or web-hosting networks. %although firewalls are implemented in software,
%there is still an incentive to minimize the number of their rules for performance reasons \cite{hipac}.
Unlike routers where TCAM puts a hard limit on the number of ACLs, there is no hard limit on the number of firewall rules, in software; however,
there is still an incentive to minimize their number and thus any associated performance penalty \cite{hipac}. %associated with increasing number of rules \cite{hipac}.

There is a body of work on firewall rule management and (mis)configuration \cite{firewalls},
%, which aims at simplifying and correcting a set of firewall rules. In particular, they
which aims at detecting anomalies such as the existence of multiple firewall rules that match the same packet, or the existence of a rule that will never match packets flowing through a specific firewall.
 %(\eg by blocking unsolicited traffic and placing protected clients behind a NAT).% or keeping state for their connections).
 % Third, the performance of a firewall degrades with the number of rules, while the forwarding performance of a router does not, as long as the ACLs are loaded in TCAM \cite{acl0}.
In contrast, we focus on resource allocation: given a blacklist and a whitelist as input to the problem, our goal is to optimally select which prefixes to filter so as to optimize an appropriate objective subject to the constraints. Furthermore, the work in \cite{firewalls} considers firewalls for enterprises, which are not supposed to be accessed from outside and thus can be protected without filtering rules.

{\bf Measurement studies.} Several measurement studies have demonstrated that malicious sources exhibit spatial and temporal clustering  \cite{uncleanness, clustering, mao2006analyzing, ramachandran2006understanding, venkataraman2007exploiting, zhang-highly}.  In order to deal with dynamic malicious IP addresses \cite{xie2007dynamic}, IP prefixes rather than individual IP addresses are typically considered. The clustering, in combination with the fact that the distribution of addresses as well as other statistical characteristics differ for good and bad traffic, have been exploited in the past for detection and mitigation of malicious traffic, such as \eg spam \cite{ramachandran2006understanding,venkataraman2007exploiting} or DDoS \cite{estan}.
In this work, we exploit these characteristics for efficient prefix-based filtering of malicious traffic.

{\bf Prefix Selection.} The work in \cite{estan} studied source prefix filtering for classification and blocking of DDoS traffic, which is closely related to our FLOODING problem. The selection of prefixes in \cite{estan} was done heuristically, thus leading to large collateral damage was incurred. In contrast, we tackle analytically the optimal source prefix selection so as to minimize collateral damage. Furthermore, we provide a more general framework for formulating and optimally solving a family of related problems, including but not limited to FLOODING.

The work in \cite{dawn}, is related to our TIME-VARYING problem:
 %and followed our conference paper in \cite{infocom09}.
   it designed and analyzed an online learning algorithm for tracking malicious IP prefixes based on a stream of labeled data. The goal was detection, \ie classifying a prefix as malicious, depending on the ratio of malicious and legitimate traffic it generates, and subject to a constraint on the number of prefixes. In contrast: (i) we identify precisely (not approximately) the IP prefixes with the highest concentration of malicious traffic; (ii) we follow a different formulation (dynamic programming inspired by knapsack problems); (iii) we use the results of detection as input to our filtering problem.

% In \cite{dawn}, the authors design an online learning algorithm
%that identifies, with high probability,  malicious regions in the IP space in a stream of labeled data. The performance of the algorithm \cite{dawn} are within a $\log$ factor from the optimal solution. This problem is related to ours as the optimal allocation of filters precisely identifies those regions of the IP space with the highest concentration of malicious traffic. However, the framework presented in our paper has the flexibility to easily accommodate other objectives as well.}

An earlier body of literature focused on identifying IP prefixes with significant amount  of network traffic, typically referred to as hierarchical heavy hitters: \cite{zhang2004online, estan2003automatically,cormode2004diamond}. However, it did not consider the interaction between legitimate the malicious traffic within the same prefix, which is the core tradeoff studied in this paper.

%\subsection{\label{sec:KP}Relation to Knapsack Problems}

{\bf Relation to Knapsack Problems.} Filter selection belongs to the family of  multidimensional knapsack problems (dKP) \cite{KPbook}.
%In dKP, there are $|\mathcal N|$ items and $d$ knapsacks, each with capacity $C_i$. Every item is associated with a profit and a weight per knapsack.
%The goal is to choose items ($x_{j}  \in \{0,1\},~\forall j \in \mathcal N$) that maximize the total profit %($max \sum_{j \in \mathcal N} p_{j}x_{j}$),
%subject to $d$ capacity constraints.% of all $d$ kanpsacks:
%%$$\sum_{j \in \mathcal N}  w_{ij} x_{j} \leq C_i, ~~ i=1,...,d$$
The general dKP problem is well-known to be NP-hard.
%\begin{align}\label{dKP}
%\max \sum_{j \in \mathcal N} p_{j}x_{j}
%\end{align}
%\vspace{-0.3cm}
%s.t.
%\vspace{-0.3cm}
%\begin{align}
%\label{dKP-capacity} 		\sum_{j \in \mathcal N}  w_{ij} x_{j} & \leq C_i && i=1,...,d\\
%\label{dKP-domain}      	x_{j} & \in \{0,1\} 	\quad && \forall i \in \mathcal N
%\end{align}
The most relevant variation is the knapsack with cardinality constraint (1.5KP) \cite{old15KP, Caprara},
which has $d=2$ constraints, one of them being a limit on the number of items:
$\sum_{j \in \mathcal N}  w_{j} x_{j}  \leq C, \sum_{j \in \mathcal N}  x_{j}  \leq k$.
%\begin{align}
%\label{1.5KP-capacity} 		\sum_{j \in \mathcal N}  w_{j} x_{j} & \leq C\\
%\label{1.5KP-cardinality} 		\sum_{j \in \mathcal N}  x_{j} & \leq k %\\
%\label{1.5KP-domain}      	x_{j} & \in \{0,1\} 	\quad && \forall i \in \mathcal N
%\end{align}
%Eq.\ref{1.5KP-cardinality} puts an upper limit to the number of items.
The 1.5KP problem is also NP-hard.%, as it reduces to KP.

These classic problems do not consider correlation between items. However, in filtering,
the selection of an item (prefix) voids the possibility to select other items (all overlapping prefixes).
dKP problems with correlation between items have been studied in \cite{dKP-GUB, 1.5KP-GUB},
where the items were partitioned into classes and up to one item per class was picked.
%In \cite{dKP-GUB}, a critical event tabu search was used to solve
%
%let  $N_1,...,N_m$ be a partition of items
%in $m$ disjoint classes. The generalized upper bound (GUB) constraint chooses at most one item per class:
%$$\sum_{j\in N_i} x_j \leq 1 ~ \forall i=1,...,m ~\text{s.t.}~  N_h \cap N_k = \emptyset, ~ \cup_i N_i = %\mathcal N$$
%
%$\begin{align}
%$\sum_{j\in N_i} x_j &\leq 1 \quad &\forall i=1,...,m \\
%$N_h \cap N_k = \emptyset, &\quad  \cup_i N_i = \mathcal N
%$\end{align}
%In \cite{1.5KP-GUB}, the same constraint has been studied with the continuous relaxation of 2KP.
In our case, a class is the set of all prefixes covering a certain address. Each item (prefix) can belong simultaneously to any number of classes, from one class (/32 address) to all classes (/0 prefix).
To the best of our knowledge, we are the first to tackle a case where the classes are not a partition of the set of items.

 A continuous relaxation does not help either. Allowing $x_{p/l}$ to be fractional corresponds to rate-limiting of prefix $p/l$. This has no advantage neither from a practical (rate limiters are more expensive than ACLs,  because in addition to looking up packets in TCAM, they also require rate and computation on the fast path) nor from a theoretical point of view (the continuous 1.5KP is still NP-hard \cite{continuous-1.5KP}.)
%%%%%%Moreover, it was shown in \cite{continuous-1.5KP} that the continuous version of KP, which admits only $k$ variables to be strictly positive, is also NP-hard.
%Relaxing %of dKP %produces a solution %in $O(n)$ time \cite{megiddo-tamir}. However,
%has an optimal solution with $\min\{\#constraints,N\}=N$
%all variables can further lead to fractional number of filters that is not meaningful.

In summary, the special structure of the prefix filtering problem, \ie the hierarchy
and overlap of candidate prefixes, leads to novel variations of dKP that could not be solved by directly applying existing methods
in the KP literature.

%\begin{theorem}
%There exist an optimal solution with at most $\min\{d,n\}$ fractional values.
%\end{theorem}
%Note: this is true for every linear program with $n$ variables, and $d$ constraint.
%In \cite{megiddo-tamir} the authors showed that the LP relaxation of dKP can be solved in $O(n)$ time.
%However, as far as filtering problems are concerned, the straightforward LP relaxation of the dKP has the
%disadvantage of labelling solutions with any number of filters (items) as feasible.

{\bf Our prior work.} This journal paper builds on our conference paper in \cite{infocom09}.
Compared to \cite{infocom09}, new contributions in this paper include: the formulation and optimal solution of the time-varying version of the filtering problem; an extended evaluation section, which simulates all filtering problems over {\tt Dshield.org} logs, including FLOODING and DIST-FLOODING which were not evaluated in \cite{infocom09}; and additional proofs, complexity analysis and comments that were not in \cite{infocom09}.
%This technical report extends our Infocom submission by presenting the details of proofs and complexity analysis that were omitted in the submission due to lack of space. In particular, the following materials in this report are new: a detailed proof of Proposition 3.1, the proof of optimality of Algorithm 1 (Theorem 3.2), the details of the complexity analysis of Algorithm 1, Proposition 3.3 and its proof, the details of the proof of Prop.3.4.

Earlier on, in a related workshop paper \cite{ita08}, we also studied optimal range-based filtering, where malicious source  addresses were aggregated into continuous ranges (of numbers in the IP address space $[0, 2^{32}-1]$), instead of prefixes. This was an easier problem that allowed for greedy solutions. Unfortunately, ranges are not implementable in ACLs; furthermore, it is well-known that ranges cannot be efficiently approximated by a combination of prefixes  \cite{varghese}. Therefore, despite the intuition we gained in \cite{ita08}, we had to solve the problem of prefix-based filtering from scratch in this paper.

\section{\label{sec:conclusion}Conclusion}

In this paper, we introduce a formal framework for optimal source prefix-based filtering. The framework is rooted at the theory of the knapsack problem and provides a novel extension to it. Within it, we formulate five practical problems, presented in increasing order of complexity. For each
problem, we designed optimal algorithms that are also low-complexity (linear or pseudo-polynomial in the input size).
%We also highlight connections between different problems: at the heart of all problems lies BLOCK-SOME; BLOCK-ALL
%and FLOODING are special instances for specific assignment of weights, and DIST-FLOODING
%decomposes into several independent FLOODING problems.
We simulate our algorithms over {\tt Dshield.org} logs and demonstrate that they bring significant benefit compared to non-optimized filter selection or to generic clustering algorithms. A key insight behind that benefit is that our algorithms exploit the spatial and temporal clustering exhibited by sources of malicious traffic.

%There are several directions for future work. We plan to extend the framework to dynamically
%update the filtering rules as blacklists change over time, combine source- with destination-based
%filtering, deal with adversarial scenarios, and study the interaction between filtering and detection
%mechanisms. We will also provide a more extensive experimental evaluation, which is not the focus
%of this paper.

%1) for the first time we introduced a formal framework to deal with
%filtering problems.
%2) our framework is rooted at the theory of the knapsack problem, and
%it provide a novel extension of it.
%3) we formulate a family of different filtering problems, and present
%them in increasing order of complexity.
%4) We provide mathematical models to those problems, and highlights
%connections between them.
%5) we design optimal and low complexity algorithms to practically
%(\ie in practical scenarios) solve them.
%6) At the hear of all filtering problems lie BLOCK-ALL and BLOCK-SOME
%problem.

\section*{Acknowledgements}
We are grateful to P. Barford and M. Blodgett at the University of Wisconsin, Madison, for making the 6-month {\tt Dshield.org} logs available to us. We would also like to thank M. Faloutsos for insightful discussions.
This work has been supported by the NSF CyberTrust grant 0831530.

\bibliographystyle{IEEEtran}
\bibliography{filtering}

\vspace{-20pt}

\begin{biography}[{\includegraphics[width=1in,height=1.25in, clip,keepaspectratio]{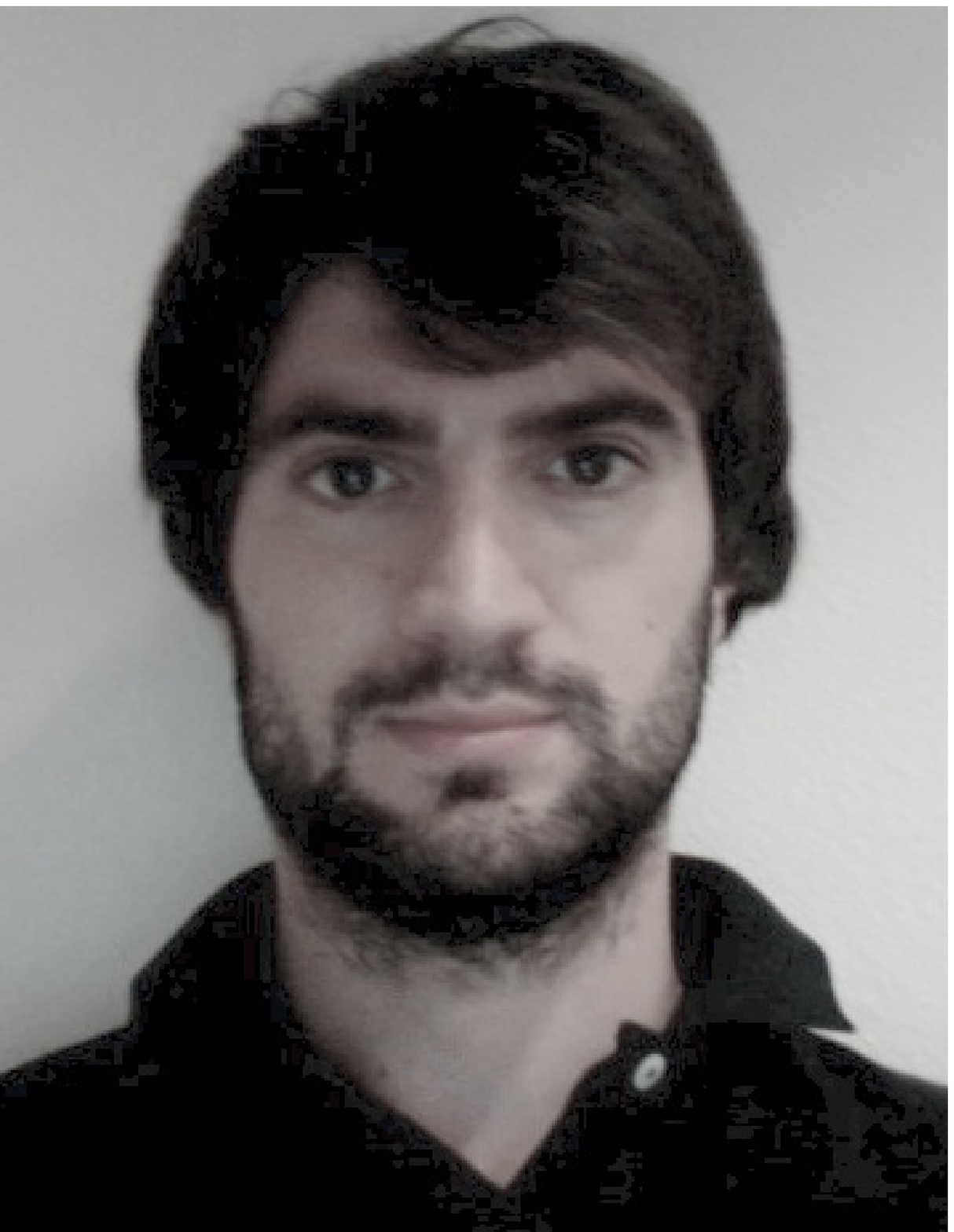}}]
{Fabio Soldo} is currently a Ph.D. candidate in the Networked Systems Program at UC Irvine. He received the B.S. degree in Mathematics from Politecnico di Torino, Italy, in 2004, and M.S. degree in Mathematical Engineering from Politecnico di Torino and Politecnico di Milano, Italy, in 2006. He had intrenships with Telefonica Research, Docomo Labs and Google. His research interests are in the areas of design and analysis of network algorithms and network protocols, and defense mechanisms against malicious traffic.
\end{biography}
\vspace{-20pt}

\begin{biography}[{\includegraphics[width=1in,height=1.25in,clip,keepaspectratio]{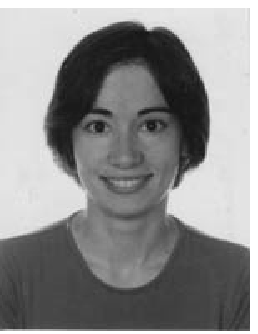}}]
{Katerina Argyraki} is a researcher with the Operating Systems group in the School of Computer
and Communication Sciences, EPFL, Switzerland.
She works on network architectures and protocols with a focus on
denial-of-service defenses and accountability.
She received her undergraduate degree in Electrical and Computer
Engineering from the Aristotle University, Thessaloniki, Greece, in 1999, and her Ph.D. in
Electrical Engineering from Stanford University, in 2007.
\end{biography}
\vspace{-20pt}

\begin{biography}[{\includegraphics[width=1in,height=1.25in,clip,keepaspectratio]{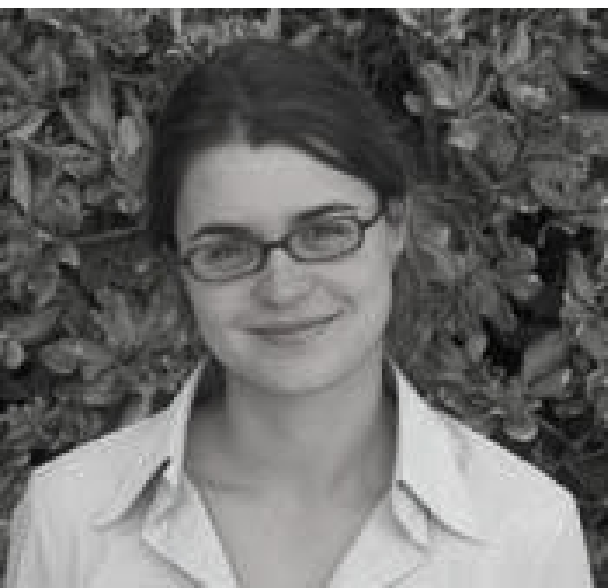}}]
{Athina Markopoulou} (SM '98, M'02) is an assistant professor in the EECS Dept. at the University of California, Irvine. She received the Diploma degree in Electrical and Computer Engineering from the National Technical University of Athens, Greece, in 1996, and the M.S. and Ph.D. degrees, both in Electrical Engineering, from Stanford University in 1998 and 2003, respectively. She has been a postdoctoral fellow at Sprint Labs (2003) and at Stanford University (2004-2005), and a member of the technical staff at Arastra Inc. (2005). Her research interests include network coding, network measurements and security, media streaming and online social networks. She received the NSF CAREER award in 2008.
\end{biography}

%\balance
\vfill

\end{document}